\documentclass{article}

\usepackage{authblk} % for affiliation
\usepackage{amsfonts,amsmath,amssymb,mathtools} % for basic math stuff
\usepackage{enumerate,hyperref} % roman enumeration and hyperref
\usepackage{filecontents} % for the bibliography
\usepackage{tikz-cd} %for diagrams
\usepackage{relsize}
\usepackage{graphicx} %for scale box
\usepackage{cleveref} % for cross-referencing

\usetikzlibrary{patterns}
\tikzcdset{every label/.append style = {font = \small}} %changes label size

\definecolor{grey}{RGB}{100,100,100}
\newtheorem{theorem}{Theorem}[subsection]
\newtheorem{lemma}[theorem]{Lemma}
\newtheorem{corollary}[theorem]{Corollary}

\newenvironment{proof}{\noindent\textit{Proof.}}{\hfill$\square \\$}

\newcommand{\Efrac}[2]{ % thanks to Werner!
  \mathchoice
    {\ooalign{%
      $\genfrac{}{}{1.2pt}0{#1}{#2}$\cr%
      $\color{white}\genfrac{}{}{.4pt}0{\phantom{#1}}{\phantom{#2}}$}}%
    {\ooalign{%
      $\genfrac{}{}{1.2pt}1{#1}{#2}$\cr%
      $\color{white}\genfrac{}{}{.4pt}1{\phantom{#1}}{\phantom{#2}}$}}%
    {\ooalign{%
      $\genfrac{}{}{1.2pt}2{#1}{#2}$\cr%
      $\color{white}\genfrac{}{}{.4pt}2{\phantom{#1}}{\phantom{#2}}$}}%
    {\ooalign{%
      $\genfrac{}{}{1.2pt}3{#1}{#2}$\cr%
      $\color{white}\genfrac{}{}{.4pt}3{\phantom{#1}}{\phantom{#2}}$}}%
}

\let\fun\lambda
\newcommand{\refl}{\mathsf{refl}}
\newcommand*\sq{\mathbin{\vcenter{\hbox{\rule{.3ex}{.3ex}}}}}
\newcommand{\htri}{\smaller\triangleright}
\newcommand{\hcom}{\mathsf{hcom}}
\newcommand{\com}{\mathsf{com}}
\newcommand{\coe}{\mathsf{coe}}

\newcommand{\filler}{\mathsf{fill}}

\newcommand{\fst}{\mathsf{fst}}
\newcommand{\snd}{\mathsf{snd}}
\newcommand{\zero}{\mathsf{0}}
\renewcommand{\succ}{\mathsf{succ}}

\let\dim\mathsf
\newcommand{\at}{@\mathsf}

\newcommand{\univ}[1]{\mathcal{#1}}

\renewcommand{\angle}[1]{\left\langle #1 \right\rangle}
\newcommand{\dangle}[1]{\angle{\dim{#1}}}
\newcommand{\dimsub}[2]{\left\langle \mathsf{#1} / \mathsf{#2} \right\rangle}
\newcommand{\lto}[2]{\mathsf{#1} \rightsquigarrow \mathsf{#2}}

\title{\bfseries{Cubical informal type theory: \\ The higher groupoid structure}}

\author[1]{Bruno Bentzen}

\affil[1]{Institute of Logic and Cognition \protect\\ Sun Yat-sen University, China}
\affil[ ]{\texttt{b.bentzen@hotmail.com}}

\date{}

\makeatletter
\newcommand\Author{Bruno Bentzen}
\let\Title\@title
\def\ps@mystyle{%
      \let\@oddfoot\@empty\let\@evenfoot\@empty
      \def\@oddhead{%
       \ifodd\value{page}\relax
          \hfill{Cubical informal type theory: the higher groupoid structure}\hfill\makebox[0pt][l]{\thepage}%
      \else
          \makebox[0pt][l]{\thepage}\hfill\Author\hfill%
      \fi%
      }%
      \let\@mkboth\markboth}
\makeatother
\begin{document}

\maketitle

%%%%%%%%%%%%%%%%%%%%%%%%%%%%%%%%
%% Abstract
%%%%%%%%%%%%%%%%%%%%%%%%%%%%%%%%

\begin{abstract}
\noindent Following a project of developing conventions and notations for informal type theory carried out in the homotopy type theory book for a framework built out of an augmentation of constructive type theory with axioms governing higher-dimensional constructions via Voevodsky's univalance axiom and higher-inductive types, this paper proposes a way of doing informal type theory with a cubical type theory as the underlying foundation instead. To that end, we adopt a cubical type theory recently proposed by Angiuli, Hou (Favonia) and Harper, a framework with a cumulative hierarchy of univalent Kan universes, full univalence and instances of higher-inductive types. In the present paper we confine ourselves to some elementary theorems concerning the higher groupoid structure of types.
\end{abstract}

%%%%%%%%%%%%%%%%%%%%%%%%%%%%%%%%
%% Text of the paper
%%%%%%%%%%%%%%%%%%%%%%%%%%%%%%%%

%%%%%%%%%%%%%%%%%%%%%%%%%%%%%%%%
\section{Introduction} \label{intro}
%%%%%%%%%%%%%%%%%%%%%%%%%%%%%%%%

Higher-dimensional type-theoretic foundations is gaining wider acceptance in the mathematical community since the emergence of homotopy type theory \cite{hottbook}, a young but promising research field and foundational language for mathematics that moves conventional type theory to higher dimensions by interpreting types as spaces, terms as points, equalities as paths and functions as continuous maps. 
One of the reasons (but certainly not the only one) for this growing interest among mathematicians can be attributed to the collective efforts of the authors of the book \emph{homotopy type theory} \cite{hottbook} to develop an informal but rigorous style of doing mathematics in natural language assuming higher-dimensional type theory as the underlying foundation. 

This `informal type theory' project, originally proposed by Peter Aczel \cite{shu13}, was carried out in the homotopy type theory book \cite{hottbook} for a framework built out of an augmentation of ordinary (one-dimensional) constructive type theory \cite{ml75} with axioms governing higher-dimensional constructions via Voevodsky's univalance axiom and higher-inductive types (henceforth `conventional homotopy type theory').
Unfortunately, however, the use of univalence and higher-inductive types as axioms is quite problematic computationally speaking, since using axioms in a type theory amounts to introducing new canonical terms without saying exactly how to compute with them. Simply put, the presence of the univalence and higher-inductive types as axioms in constructive type theory blocks computation, meaning that conventional homotopy type theory lacks all the desirable computational properties of a type theory such as canonicity \cite{hottbook}.

In response to that, Bezem et al. \cite{bch14} have constructed a model of constructive type theory that validates the univalence axiom using cubical sets (a constructive mathematical concept due to Kan \cite{kan55}) and many cubical type theories (type-theoretic paraphrases of the cubical interpretation) have been developed since then.\footnote{Not all recently developed cubical type theories \cite{cch16,lb14,ah17} are based on the same sort of cubical structure \cite{bmx}. The version of cubical sets which is used in this paper, for example, is not quite the same as Kan \cite{kan55}, since our cubical sets are symmetric and the use of symmetry is essential to ensure that we have a symmetric tensor product.} 
Cohen et al. \cite{cch16} have proposed a cubical type theory which proves univalence and has possible extensions with some higher inductive types. Licata and Brunerie \cite{lb14} have introduced a cubical type theory with instances of higher-inductive types and, very recently, Angiuli et al. \cite{chtt3} have presented a cubical type theory with a cumulative hierarchy of univalent universes, full univalence and instances of higher-inductive types. 

The aim of this paper is to offer a cubical perspective to the informal type theory of the homotopy type theory book by adopting not conventional homotopy type theory but a cubical type theory as the implicit basis of our informal reasoning. Our approach is based on the framework of Angiuli et al. \cite{chtt3}, which we shall refer to as `computational cubical type theory'. Thus, although this paper is intended to be self-contained, the reader may find it helpful to refer to \cite{chtt3,ang17,ahx} for further clarification.

%%%%%%%%%%%%%%%%%%%%%%%%%%%%%%%%
\section{Computational cubical type theory} \label{cctt}
%%%%%%%%%%%%%%%%%%%%%%%%%%%%%%%%

As usual in type theory, the language of computational cubical type theory is composed of ordinary terms from an extended lambda calculus with constants for the constructors and eliminators of type formers, such as

$$\fun x.M, \; M(N), \; \dangle{M,N}, \; \fst (M), \; \snd (M),\; \zero,\; \succ(M),\; ... $$
However, computational cubical type theory features a very unique sort of terms, called \textit{dimension terms}, which can be combined with ordinary terms to provide an explicit higher-dimensional treatment of the terms of the language. 

%%%%%%%%%%%%%%%%%%%%%%%%%%%%%%%%
\subsection{Dimension terms}
%%%%%%%%%%%%%%%%%%%%%%%%%%%%%%%%

What are exactly dimension terms? Syntactically, a dimension term is either $\mathsf{0}$ or $\mathsf{1}$ (which we sometimes abbreviate as $\mathsf{\epsilon}$) or a dimension name: \textsf{x}, \textsf{y}, \textsf{z}, ... (which we always write in \textsf{sans-serif}). Semantically, we can think of a dimension term as an abstract point given in a type-theoretical representation of the unit interval space $\left[ \mathsf{0}, \mathsf{1} \right]$. 
Just like types and their terms may contain free variables, they may also contain \textit{dimension names}. We say that types with no occurrence of dimension names are at dimension zero, which are just types in the traditional sense. We also say that a type with exactly one, two, three, ..., $n$ dimension names are respectively at dimension one, two, three, ..., $n$.
It is often helpful to mention the dimension names contained in a type explicitly. For example, if a type at dimension one contains exactly one dimension name $\dim{x}$ we may call it an $\dim{x}$-type. We may also refer to a type at dimension two containing exactly two dimension names $\dim{x}$ and $\dim{y}$ as an $(\dim{x},\dim{y})$-type, a type at dimension three containing exactly three dimension names $\dim{x}$, $\dim{y}$, $\dim{z}$ as an $(\dim{x},\dim{y},\dim{z})$-type and so on (types at dimension zero can be called $0$-types, or simply types for short, since 0-types are just types in the sense of conventional one-dimensional type theory).

A crucial feature of dimension names is substitution: given any term $M$, any dimension name $\dim x$, and any dimension term $\dim r$, we have a dimension substitution operation $M \dimsub{r}{x}$ which replaces all occurrences of $\dim x$ in $M$ with $\dim r$ (note that $\dim r$ may be either a dimension name or $\mathsf{\epsilon}$).\footnote{Since computational cubical type theory possesses a universe of types (including a cumulative hierarchy of univalent Kan universes \cite{chtt3}), dimension substitution is a well-defined operation for types (regarded as terms in a universe) as well.}

Dimension substitution allows types and terms at arbitrary dimensions to be characterized as follows. 
The trivial case is, of course, that of types and terms at dimension zero. In this context, every type $A$ represents a type point and, if $M$ is a term that belongs to $A$ (in which case we shall write $M \in A$), we say that $M$ is a point in $A$. 
In the one-dimensional case, any $\dim{x}$-type $A$ can be seen as a type line from the (zero-dimensional) type $A \dimsub{0}{x}$ to $A \dimsub{1}{x}$, and, if $M \in A$, then $M$ represents an $\dim{x}$-line in $A$ from $M \dimsub{0}{x} \in A \dimsub{0}{x}$ to $M \dimsub{1}{x} \in A \dimsub{1}{x}$. When drawing Kan composition diagrams (in the sense described in item (\ref{hcom}) of \Cref{hd-op}), lines will be often illustrated as follows:

\vspace{2mm}
\begin{tikzcd}
{} \arrow[r,shorten >= 12pt,-latex,start anchor=center,"\dim x" near start] & M \dimsub{0}{x}
       \arrow[rrrrrrrr,"M"] &&&&&&&& M \dimsub{1}{x}
\end{tikzcd}
\vspace{2mm}

\noindent Two-dimensionally speaking, we can think of types and their inhabitants in terms of squares. 
So an $(\dim{x},\dim{y})$-type $A$ can be seen as a type square with, respectively, the $\dim x$-types $A \dimsub{0}{y}$ and $A \dimsub{1}{y}$ at the top and bottom, and, respectively, the $\dim y$-types $A \dimsub{0}{x}$ and $A \dimsub{1}{x}$ at the left and right. When $M \in A$ we say that $M$ is an $(\dim{x},\dim{y})$-square in $A$ with, respectively, the $\dim x$-lines $M \dimsub{0}{y} \in A \dimsub{0}{y}$ and $M \dimsub{1}{y} \in A \dimsub{1}{y}$ at the top and bottom, and, respectively, the $\dim y$-lines $M \dimsub{0}{x} \in A \dimsub{0}{x}$ and $M \dimsub{1}{x} \in A \dimsub{1}{x}$ at the left and right. 
The following diagram summarizes the above construction:

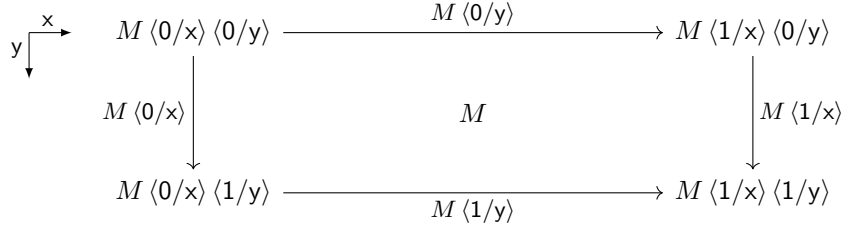
\vspace{2mm}
\begin{tikzcd}
{} \arrow[r,shorten >= 12pt,-latex,start anchor=center,"\dim x" near start] \arrow[d,shorten >= 9pt,-latex,swap,start anchor=center,"\dim y" near start] & M \dimsub{0}{x}\dimsub{0}{y}
       \arrow[rrrrrr,"M \dimsub{0}{y}"] \arrow[dd,swap,"M \dimsub{0}{x}"] \arrow[ddrrrrrr, phantom, "M"] &&&&&& M \dimsub{1}{x} \dimsub{0}{y} \arrow[dd,"M \dimsub{1}{x}"] \\
 {} &   & \\
&   M \dimsub{0}{x}\dimsub{1}{y} \arrow[rrrrrr, swap,"M \dimsub{1}{y}"] &&&&&& M \dimsub{1}{x}\dimsub{1}{y}
\end{tikzcd}
\vspace{2mm}

\noindent 
Any three-dimensional type such as an $(\dim{x},\dim{y},\dim{z})$-type $A$ represents a type cube with the $(\dim x,\dim z)$-types $A \dimsub{0}{y}$ and $A \dimsub{1}{y}$ respectively at the top and bottom, the $(\dim z,\dim y)$-types $A \dimsub{0}{x}$ and $A \dimsub{1}{x}$ respectively at the left and right and $(\dim x,\dim y)$-types $A \dimsub{0}{z}$ and $A \dimsub{1}{z}$ respectively at the front and back. 
When $M \in A$ we have an $(\dim x,\dim y,\dim z)$-cube $M$ in $A$ with the $(\dim x,\dim z)$-squares $M \dimsub{0}{y}$ and $M \dimsub{1}{y}$ respectively at the top and bottom, the $(\dim z,\dim y)$-squares $M \dimsub{0}{x}$ and $M \dimsub{1}{x}$ respectively at the left and right and $(\dim x,\dim y)$-squares $M \dimsub{0}{z}$ and $M \dimsub{1}{z}$ respectively at the front and back. 
Such cubes will be represented as follows:

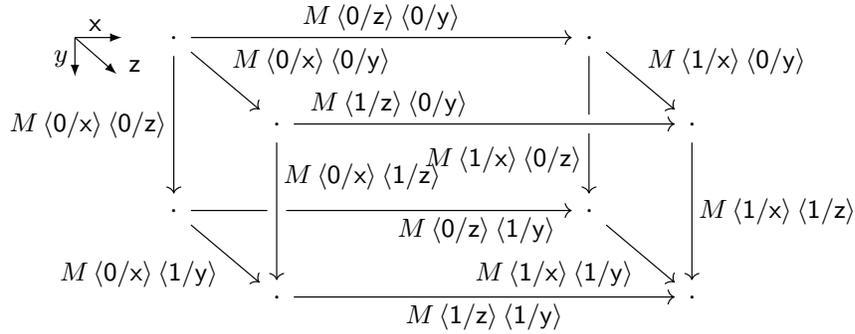
\vspace{2mm}

\begin{tikzpicture}[baseline= (a).base]
\node[scale=1.1] (a) at (0,0){

\begin{tikzcd}
{} \arrow[r,shorten >= 12pt,-latex,start anchor=center,"\dim x" near start] \arrow[dr,shorten >= 24pt,-latex,start anchor=center,"\dim z"] \arrow[d,shorten >= 15pt,-latex,swap,start anchor=center,"y" near start] & {}
%&%&
  \cdot \arrow[rrrrr,"M \dimsub{0}{z}\dimsub{0}{y}"] \arrow[dr,"M \dimsub{0}{x}\dimsub{0}{y}"] \arrow[dd,swap,"M \dimsub{0}{x}\dimsub{0}{z}"] \arrow[drrrrrr, phantom] && &&&
                   \cdot \arrow[dd,swap,"M \dimsub{1}{x}\dimsub{0}{z}" near end] \arrow[dr,"M \dimsub{1}{x}\dimsub{0}{y}"] \\
{} & {} 
& \cdot \arrow[rrrrr,crossing over,"M \dimsub{1}{z}\dimsub{0}{y}" near start] &&    &&& \cdot \arrow[dd,"M \dimsub{1}{x}\dimsub{1}{z}"] \\
& \cdot \arrow[rrrrr,swap,"M \dimsub{0}{z}\dimsub{1}{y}" near end] \arrow[dr,swap,"M \dimsub{0}{x}\dimsub{1}{y}"] &&    &&& \cdot \arrow[dr,swap,"M \dimsub{1}{x}\dimsub{1}{y}"] \\
& & \cdot \arrow[rrrrr,swap,"M \dimsub{1}{z}\dimsub{1}{y}"] \arrow[uu,<-,crossing over,swap,"M \dimsub{0}{x}\dimsub{1}{z}" near end] &&  &&& \cdot
\end{tikzcd}

};
\end{tikzpicture}

\vspace{2mm}

\noindent Observe that the labels of the vertices of the cube displayed in the above diagram have been omitted for simplicity. This information, however, can be easily inferred from the labels of its edges. For example, we know that the top-left-back vertex of this cube must be strictly equal to both the right boundary of $M \dimsub{0}{z}\dimsub{0}{y}$ and the top boundary of $M \dimsub{0}{x}\dimsub{0}{z}$, so this cube's top-left-back vertex must be $M \dimsub{0}{z}\dimsub{0}{y} \dimsub{0}{x} \equiv M \dimsub{0}{x}\dimsub{0}{z} \dimsub{0}{y}$ and so on.

%%%%%%%%%%%%%%%%%%%%%%%%%%%%%%%%
\subsection{Higher-dimensional operations} \label{hd-op}
%%%%%%%%%%%%%%%%%%%%%%%%%%%%%%%%

Before the introduction of the identification type (the type of identifications between two inhabitants of a type), it is important to mention the main three higher-dimensional operations of computational cubical type theory:

\begin{enumerate}
 \item \textbf{Degeneration}. This allows one to trivially regard a construction at any dimension as a higher construction at the next dimension. For instance, any $(\dim x_1,...,\dim x_n)$-cube $M$ can be degenerated into an $(\dim x_1,...,\dim x_n, \dim x)$-cube $M$ with a trivial $\dim x$ face, that is, we have $A \dimsub{0}{x} \equiv A \dimsub{1}{x} \equiv A$.
When drawing Kan composition diagrams, we shall always use double lines ($=\joinrel=\joinrel=$) to indicate degenerate faces.

	\item \textbf{Coercion}. This can be seen as a cubical generalization of the transport operation from the homotopy type theory book \cite[Lem 2.3.1]{hottbook}. Essentially, coercion states that, given any $\dim x$-type $A$ and any term $M \in A \dimsub{r}{x}$, we have a term of the type $A \dimsub{r'}{x}$, called the coercion of $M$ in $A$, and denoted by $\coe^{\lto{r}{ r'}}_{\dim{x}.A} (M)$. \label{coe}
	
	\item \textbf{Homogeneous Kan composition}. Simply put, homogeneous Kan composition ensures that any open box has a lid. The simplest composition scenario can be illustrated as follows: \label{hcom}
	
\end{enumerate}

\vspace{2mm}
\begin{tikzcd}
{} \arrow[r,shorten >= 12pt,-latex,start anchor=center,"\dim x" near start] \arrow[d,shorten >= 9pt,-latex,swap,start anchor=center,"\dim y" near start] & N_0 \dimsub{0}{y} \equiv M \dimsub{0}{x}
       \arrow[rrrrr,"M"] \arrow[dd,swap,"N_0"] \arrow[ddrrrrr, phantom, "{\hcom^{\lto{0}{y}}_{\dim A}(M)({}^\dim{0}_\dim{x} \htri \dim y.N_0, \; {}^\dim{1}_\dim{x} \htri \dim y.N_1)}"] &&&&& M \dimsub{1}{x} \equiv N_1 \dimsub{0}{y} \arrow[dd,"N_1"] \\
 {} &   & \\
&   N_0 \dimsub{1}{y} \arrow[rrrrr,dotted, swap,"{\hcom^{\lto{0}{1}}_{\dim A}(M)({}^\dim{0}_\dim{x} \htri \dim y.N_0, \; {}^\dim{1}_\dim{x} \htri \dim y.N_1)}"] &&&&& N_1 \dimsub{1}{y}
\end{tikzcd}
\vspace{2mm}

\begin{itemize}
	\item[{}] The above diagram states that given any $\dim{x}$-line $M$ in a type $A$ and two $\dim{y}$-lines $N_0$ and $N_0$ in $A$ such that (i) the left boundary of $M$ is strictly equal to the left boundary of $N_0$ and (ii) the right boundary of $M$ is strictly equal to the left boundary of $N_0$, there exists an $\dim x$-line in $A$ from the right boundary of $N_0$ to the right boundary of $N_1$ (the dotted line in the diagram). The resulting new line is called the homogeneous Kan composite of $M$ with $N_0$ and $N_1$ and denoted by $\hcom^{\lto{0}{1}}_{\dim A}(M)({}^\dim{0}_\dim{x} \htri \dim y.N_0, \; {}^\dim{1}_\dim{x} \htri \dim y.N_1)$. 
Crucially, homogeneous Kan composition also asserts the existence of an $(\dim x,\dim y)$-square in $A$ with the $\dim x$-lines $M$ and $\hcom^{\lto{0}{1}}_{\dim A}(M)({}^\dim{0}_\dim{x} \htri \dim y.N_0, \; {}^\dim{1}_\dim{x} \htri \dim y.N_1)$ respectively at the top and bottom, and the $\dim y$-lines $N_0$ and $N_1$ respectively at the left and right (this is the square depicted in the above diagram). We call the resulting square the \emph{filler} of the Kan composition scenario.
	
It is worth mentioning that homogeneous Kan composition need not be limited to two-dimensional open boxes. In fact, its general form is
$$\hcom^{\lto{r}{s}}_{\dim A}(M)({}^\dim{0}_\dim{x_1} \htri \dim y.N_{\dim{x_1}0}, \; {}^\dim{1}_\dim{x_1} \htri \dim y.N_{\dim{x_0}1}, ... , \; {}^\dim{0}_\dim{x_n} \htri \dim y.N_{\dim{x_n}0}, \; {}^\dim{1}_\dim{x_n} \htri \dim y.N_{\dim{x_n}1}).$$
\end{itemize}

The constructs from items \ref{coe} and \ref{hcom} are called the \emph{Kan conditions}, and together they can be seen as a higher-dimensional representation of the elimination rule of the identification type of constructive type theory known in the homotopy type book as path induction \cite[\S 1.12.1]{hottbook} (see \Cref{pathind}).

%%%%%%%%%%%%%%%%%%%%%%%%%%%%%%%%
\subsection{Identification type}
%%%%%%%%%%%%%%%%%%%%%%%%%%%%%%%%

Given any $\dim x$-type $A$ and any two terms $M \in A \dimsub{0}{x}$ and $N \in A \dimsub{1}{x}$ we can construct the type $M =_{\dim x.A} N$ of identifications between the terms $M$ and $N$ in the type $A$ indexed by $\dim x$. It is important to emphasize that, unlike in conventional homotopy type theory \cite{hottbook}, the identification type of computational cubical type theory is always indexed by a particular dimension term.

The introduction, elimination, computation and uniqueness rules of the identification type are the following:

\begin{enumerate}
	\item \textbf{Introduction}. The identification type is inhabited by identifications, which are constructed by dimension abstraction. Given any $\dim x$-line $P$ in $A$ from $M$ to $N$, we write $\dangle{x} P$ to indicate the identification of $M$ and $N$ in $A$ obtained by abstracting $\dim x$ in the $\dim x$-line $P$. Consequently, all occurrences of $\dim x$ in $P$ are binded in the identification $\dangle{x} P$, and, because the resulting identification does not depend on the dimension term $\dim x$, as a general rule, identifications formed by abstracting lines ($n+1$-cubes) in a type can be seen as points ($n$-cubes) in their corresponding identification types.
	
	\item \textbf{Elimination}. Given any identification $P \in M_0 =_{\dim x.A} M_1$ and a dimension term $\dim r$, we can apply the identification $P$ to $\dim r$ to obtain an $\dim r$-line from $M_0$ to $N_1$ in $A \dimsub{r}{x}$, denoted $P \at r$.
	We also require that $P \at \mathsf{\epsilon} \equiv M_{\mathsf{\epsilon}} \in A \dimsub{\mathsf\epsilon}{x}$.
	
	\item \textbf{Computation}. We allow bound dimension names to be used interchangeably and require that all terms obtained by dimension abstraction always lead to certain lines when applied to a certain dimension terms. 
	It is often convenient to express those verbose conditions with the following $\alpha$- and $\beta$-rules:

\begin{itemize}
	\item[($\alpha$)] $\dangle{x} M \equiv \dangle{y} (M \dimsub{y}{x})$;
	\item[($\beta$)] $(\dangle{x} M)@r \equiv M \dimsub{r}{x}$.
\end{itemize}

	\item \textbf{Uniqueness}. We also endorse an extensional view of identifications, which is to say that we require the following $\eta$-rule to hold:

\begin{itemize}
		\item[($\eta$)] $\dangle{x} (M\at x) \equiv M$ (when $\dim x$ does not occur in $M$.)
\end{itemize}

\end{enumerate}
It goes without saying that this type has many similarities with the function type (except that it deals with dimension names and the function type variables).

It worth noting, however, that this type introduces an heterogeneous approach to equality that is fundamentally different from the homogeneous account found in conventional homotopy type theory. 
This is because the terms $M \in A \dimsub{0}{x}$ and $N \in A \dimsub{1}{x}$ from any identification type $M =_{\dim x.A} N$ need not share the same type, since we have $A \dimsub{0}{x} \not\equiv A \dimsub{1}{x}$ in general (but not when $A$ is a degenerate $\dim x$-type). We shall investigate heterogeneous equality in more details in \Cref{heterogroupstruct}.

%%

%%%%%%%%%%%%%%%%%%%%%%%%%%%%%%%%
\section{Higher groupoid structure} \label{hgroup}
%%%%%%%%%%%%%%%%%%%%%%%%%%%%%%%%

We begin our account with the central idea of higher type theory, namely, that types can be regarded as (weak) higher groupoids, a category in which all morphisms are isomorphisms up to a higher morphism. When types are regarded as higher groupoids, identifications can be seen as morphisms. This means that we need to define reflexivity, inversion and composition operators for identifications and show that they are well-behaved in a sense that will be explained in \Cref{hgroupstruct}.
 
%%%%%%%%%%%%%%%%%%%%%%%%%%%%%%%%
\subsection{Homogeneous groupoid operations} \label{hop}
%%%%%%%%%%%%%%%%%%%%%%%%%%%%%%%%

We start with the definition of our identity element, the cubical counterpart of the reflexivity identification from conventional homotopy type theory \cite{hottbook}:

%%%%

\begin{lemma}[Reflexivity] \label{def:refl}
For every degenerated $\dim x$-type $A$ and every $a \in A \dimsub{0}{x}$, there exists an identification of $a$ and $a$ in $A$
$$a=_{\dim{x}.A}a$$
called the reflexivity identification of $a$ and denoted $\refl_a$.
\end{lemma}
\begin{proof} By assumption, $A$ is a degenerated $\dim x$-type, so we have $A \dimsub{0}{x} \equiv A \dimsub{1}{x}$. Thus, we have both $a \in A \dimsub{0}{x}$ and $a \in A \dimsub{1}{x}$ and the type $a=_{\dim{x}.A}a$ is well-formed. Because degeneracy allow us to regard $a$ as an $\dim x$-line from $a \in A \dimsub{0}{x}$ to $a \in A \dimsub{1}{x}$ at $A$ (a degenerated line), we simply define $\refl_a :\equiv \dangle{x} a$.
\end{proof}

%%%%

Now that we have a well-defined notion of our identity element we start our cubical constructions with our preliminary definitions of symmetry (inversion) and transitivity (composition) of identifications. Let us first consider the former:

%%%%

\begin{lemma}[Inversion] \label{lem:inv}
For every degenerated $\dim x$-type $A$ and every $a \in A \dimsub{0}{x}$ and $b \in A \dimsub{1}{x}$, there is a function
$$(a=_{\dim{x}.A}b) \to (b=_{\dim{x}.A}a)$$
called the inverse function and denoted $p \mapsto p^{-1}$.
\end{lemma}
\begin{proof}
As before, we have $A \dimsub{0}{x} \equiv A \dimsub{1}{x}$. Suppose that $\dim y$ is a dimension term. Degeneracy allows us to regard $A$ as an $(\dim x, \dim y )$-type and to infer that $A \dimsub{0}{y} \equiv A \dimsub{1}{y}$ as $\dim x$-types (since $A$ is an $\dim x$-type by assumption). Moreover, the types $A \dimsub{0}{x}\dimsub{0}{y}$, $A \dimsub{0}{x}\dimsub{1}{y}$, $A \dimsub{1}{x}\dimsub{0}{y}$, $A \dimsub{1}{x}\dimsub{1}{y}$ are all strictly equal, so we may use them interchangeably in all contexts.

The idea of the following proof is to observe that, since $p \in a=_{\dim{x}.A}b$, we have that $p \at y$ is a $\dim y$-line from $a \in A \dimsub{0}{x} \dimsub{0}{y}$ to $b \in A \dimsub{0}{x}\dimsub{1}{y}$ in $A$, now trivially regarded as an $(\dim x, \dim y )$-type. Similarly, $\refl_a \in a=_{\dim{x}.A}a$ gives us a (degenerated) $\dim x$-line $\refl_a \at x$ from $a \in A \dimsub{0}{x}$ to $a \in A \dimsub{1}{x}$ in $A$ and a (degenerated) $\dim y$-line $\refl_a \at y$ from $a \in A \dimsub{0}{y}$ to $a \in A \dimsub{1}{y}$ in $A$. 

Now we note that the left boundaries of $p\at{y}$ and $\refl_a \at{x}$ are both strictly equal, and that the right boundary $\refl_a \at x$ and the left boundary of $\refl_a \at  y$ are both strictly equal as well (all those boundaries are $a$). In other words, we have an open square whose faces are formed by the lines $p\at y$ (right), $\refl_a \at x$ (top) and $\refl_a \at y$ (left). By homogeneous Kan composition, this open square must have a lid (bottom), so we have an $\dim x$-line from $b \in A \dimsub{0}{x}\dimsub{1}{y}$ to $a \in A \dimsub{1}{x}\dimsub{1}{y}$ in $A$, as illustrated as the dotted line in the diagram below:

\vspace{2mm}
\begin{tikzcd}
{} \arrow[r,shorten >= 12pt,-latex,start anchor=center,"\dim x" near start] \arrow[d,shorten >= 9pt,-latex,swap,start anchor=center,"\dim y" near start] & a
       \arrow[rrrrrrrr,equal,"\refl_a\at x"] \arrow[dd,swap,"p\at y"] &&&&&&&& a \arrow[dd,equal,"\refl_a\at y"] \\
 {} &   & \\
&   b \arrow[rrrrrrrr,dotted, swap,"{\hcom^{\lto{0}{1}}_{\dim A}(\refl_a\at x)({}^\dim{0}_\dim{x} \htri \dim y.p\at y, \; {}^\dim{1}_\dim{x} \htri \dim y. \refl_a\at y)}"] &&&&&&&& a
\end{tikzcd}
\vspace{2mm}

\noindent Thus, we define $p^{-1} :\equiv \dangle{x} \hcom^{\lto{0}{1}}_{\dim A}(\refl_a\at x)({}^\dim{0}_\dim{x} \htri \dim y.p\at y, \; {}^\dim{1}_\dim{x} \htri \dim y. \refl_a\at y)$. 
\end{proof}

%%%%

Recall that since $p^{-1}\at x$, that is, $\hcom^{\lto{0}{1}}_{\dim A}(\refl_a\at x)({}^\dim{0}_\dim{x} \htri \dim y.p\at y, \; {}^\dim{1}_\dim{x} \htri \dim y. \refl_a\at y)$, represents the homogeneous Kan composition of the diagram depicted above, the $(\dim x,\dim y)$-square in the diagram above is witnessed by the filler $$\hcom^{\lto{0}{y}}_{\dim A}(\refl_a\at x)({}^\dim{0}_\dim{x} \htri \dim y.p\at y, \; {}^\dim{1}_\dim{x} \htri \dim y. \refl_a\at y).$$%%
Because not only Kan compositions but also their corresponding fillers will be extremely relevant to our constructions later on, it is useful to have a special symbolism to talk about them in a more convenient way. This motivates the following notation: if $M$ stands for the Kan composition of an open box, then $\filler_{\dim y}(M)$ will stand for the filler of the Kan composition scenario in the dimension $\dim y$. For example, in our above proof of Lemma \ref{lem:inv} where $p^{-1}\at x$ refers to the Kan composition of the open box, the filler that witnesses the above $(\dim x,\dim y)$-square can be denoted by $\filler_{\dim y}(p^{-1}\at x)$.

Just as we constructed our inversion operation using homogeneous Kan composition, we can define our preliminary notion of composition of identifications in a similar way:

%%%%

\begin{lemma}[Composition] \label{lem:comp}
For every degenerated $\dim x$-type $A$ and every $a \in A \dimsub{0}{x}$, $b \in A \dimsub{1}{x}$ and $c \in A \dimsub{1}{x}$, there is a function
$$ (a=_{\dim x.A}b) \to (b=_{\dim x.A}c) \to (a=_{\dim x.A}c) $$
denoted $ p \mapsto q \mapsto p \sq q$. We call $p \sq q$ the composition of $p$ and $q$.
\end{lemma}
\begin{proof} Once again, we assume that $\dim y$ is a dimension term so that the types $A \dimsub{0}{x}\dimsub{0}{y}$, $A \dimsub{0}{x}\dimsub{1}{y}$, $A \dimsub{1}{x}\dimsub{0}{y}$, $A \dimsub{1}{x}\dimsub{1}{y}$ are all strictly equal. 

Given the identifications $p \in a=_{\dim x.A}b$ and $q \in b=_{\dim x.A}c$, we can construct three lines in $A$: an $\dim x$-line $p\at x$ from $a \in A \dimsub{0}{x}\dimsub{0}{y}$ to $b \in A \dimsub{1}{x}\dimsub{0}{y}$, a $\dim y$-line $q\at y$ from $b \in A \dimsub{1}{x}\dimsub{0}{y}$ to $c \in A \dimsub{1}{x}\dimsub{1}{y}$ and a (degenerate) $\dim y$-line $\refl_a\at y$. It is easy to see that the left boundaries of $\refl_a\at y$ and $p\at x$ and that the right boundary of $p\at x$ and the left boundary of $q\at y$ are all strictly equal. Thus, again, we have an open square, as indicated in the following diagram:

\vspace{2mm}
\begin{tikzcd}
{} \arrow[r,shorten >= 12pt,-latex,start anchor=center,"\dim x" near start] \arrow[d,shorten >= 9pt,-latex,swap,start anchor=center,"\dim y" near start] &
    a  \arrow[rrrrrrrr,"p\at x"] \arrow[dd,swap,equal,"\refl_a\at y"] \arrow[ddrrrrrrrr, phantom] &&&&&&&& b \arrow[dd,"q\at y"] \\
 {} &   & \\
&   a \arrow[rrrrrrrr,dotted, swap,"{\hcom^{\lto{0}{1}}_{\dim A}(p\at x)({}^\dim{0}_\dim{x} \htri \dim y.\refl_a\at y, \; {}^\dim{1}_\dim{x} \htri \dim y. q\at y)}"] &&&&&&&& c
\end{tikzcd}
\vspace{2mm}

\noindent Since we can construct an $\dim x$-line from $a \in A \dimsub{0}{x}\dimsub{1}{y}$ to $c \in A \dimsub{1}{x}\dimsub{1}{y}$ by homogeneous Kan composition, the construction
$$ p \sq q :\equiv \dangle{x} \hcom^{\lto{0}{1}}_{\dim A}(p\at x)({}^\dim{0}_\dim{x} \htri \dim y.\refl_a\at y, \; {}^\dim{1}_\dim{x} \htri \dim y. q\at y)$$
gives us the required identification of $a$ and $c$ in $A$.
\end{proof}

%%%%

We shall try to make our propositions as explicit as possible throughout the remainder of this paper, but, for the sake of readability,  we shall often omit labels for degenerated lines when drawing filling diagrams from now on: this information is always irrelevant since the reader should be able to correctly guess the label of any given degenerated line by checking its endpoints. 
We may also often omit assumptions about dimension terms and treat $\alpha$-, $\beta$- and $\eta$-conversions of identifications implicitly (for example, since $\refl_{M} \at r :\equiv (\dangle{x} M) \at r$ always induces a $\beta$-conversion, we shall often use the terms $\refl_{M} \at r$ and $M$ interchangeably without further comment).
When proving a proposition we may also generally refer to previous propositions (writing e.g. `by Lemma X') if we trust that the reader is able to insert the correct instances of it. 

%%%%%%%%%%%%%%%%%%%%%%%%%%%%%%%%
\subsection{Homogeneous groupoid structure} \label{hgroupstruct}
%%%%%%%%%%%%%%%%%%%%%%%%%%%%%%%%

Now that we have a well-defined reflexivity element and inverse and composition operations we need to know if they are well-behaved in the sense that they respect the (weak) higher groupoid structure (up to a higher identification).\footnote{For a detailed account of the homotopy interpretation of type theory see \cite{hottbook,aw09}.} More specifically, we need to make sure that the reflexivity element is a unit for inversion and composition, that inversion indeed provides inverses and that composition is associative. 
For now we shall focus on the first claim (the other ones will be made fully precise later).

What does it mean to say that the reflexivity element is a unit for inversion? The higher groupoid structure only holds up to higher identification, so this means that the reflexivity element should equal its inverse up to higher identification. What is the most general non-trivial higher identification in this case? 
Recall that the reflexivity element can be regarded as a one-dimensional identification, so we can think of this higher identification as a two-dimensional identification that simultaneously identifies two pairs of one-identifications. 
Given these points, it becomes clear that the answer is in a very particular higher identification that not only identifies the reflexivity element with its inverse but also simultaneously identifies two \textit{degenerate lines} (which is merely another word for reflexivity).

In this case, this very special sort of identification (henceforth, `globular identification') can be pictured as an $(\dim x, \dim y)$-square that has the two particular lines (the reflexivity element and its inverse) identified vertically as $\dim x$-lines and two degenerate lines on the other two $\dim y$-sides (intuitively, globular identifications are just generalized lines).
Accordingly, the following lemma can be stated as follows:

%%%%

\begin{lemma}[Inversion unit] \label{lem:invu}
For every degenerate $\dim x$-type $A$ and every $a \in A \dimsub{0}{x}$, we have an identification
$$\mathsf{iu}_a \in \refl_a =_{\dim y. a =_{\dim x.A} a} {\refl_a}^{-1}.$$
\end{lemma}
\begin{proof}
We construct an $(\dim x, \dim y)$-square that simultaneously identify $\refl_a \at x$ with ${\refl_a}^{-1} \at x$ as $\dim x$-lines and $\refl_a \at y$ with $\refl_a \at y$ as $\dim y$-lines. This, however, follows immediately from Lemma \ref{lem:inv}, since the filler of the inverse of ${\refl_a}^{-1}$ witnesses the $(\dim x, \dim y)$-square

\vspace{2mm}
\begin{tikzcd}
{} \arrow[r,shorten >= 12pt,-latex,start anchor=center,"\dim x" near start] \arrow[d,shorten >= 9pt,-latex,swap,start anchor=center,"\dim y" near start] &
    a  \arrow[rrrrrrrr,equal,"\refl_a\at x"] \arrow[dd,equal] \arrow[ddrrrrrrrr, phantom, "\filler_{\dim y}({\refl_a}^{-1}\at x)"] &&&&&&&& a \arrow[dd,equal] \\
 {} &   & \\
&   a \arrow[rrrrrrrr, swap,"(\refl_a)^{-1}\at x"] &&&&&&&& a
\end{tikzcd}
\vspace{2mm}

\noindent We let $ \mathsf{iu}(a) :\equiv \dangle{y}\dangle{x} \filler_{\dim y}({\refl_a}^{-1}\at x)$ be the required identification.
\end{proof}

%%%%

Similarly, composition has no effect on the reflexivity element either, for the reflexivity elements equals the composition of the reflexivity element with itself up to globular identification.

%%%%

\begin{lemma}[Composition unit] \label{lem:compu}
For every degenerate $\dim x$-type $A$ and every $a \in A \dimsub{0}{x}$, we have an identification
$$\mathsf{cu}_a \in \refl_a =_{\dim y. a =_{\dim x.A} a} \refl_a \sq \refl_a.$$
\end{lemma}
\begin{proof}
The proof is straightforward using the filler of the composition of reflexivity $\refl_a$ with $\refl_a$ from Lemma \ref{lem:comp}.
\end{proof}

%%%%

We now wish to show that the reflexivity element is a right and left unit for composition up to globular identification. Because the proof is simpler for the right unit, we shall consider it first.

\begin{lemma}[Right unit] \label{lem:ru}
For every degenerate $\dim x$-type $A$ and every $a \in A \dimsub{0}{x}$ and $b \in A \dimsub{1}{x}$ we have an identification
$$\mathsf{ru}_p \in p =_{\dim y.a =_{\dim x.A} b} p \sq \refl_b$$
for any $p \in a =_{\dim x.A} b$.
\end{lemma}
\begin{proof} We need to construct an identification of $p$ and $p \sq \refl_b$ in the identification type $a =_{\dim x.A}b$, or, in other words, an $(\dim x,\dim y)$-square having $p \at x$ and $(p \sq \refl_b) \at x$ as $\dim x$-lines and $\refl_a \at y$ and $\refl_b \at y$ as degenerate $\dim y$-lines. 
But the existence of this square follows from Lemma \ref{lem:comp}:

\vspace{2mm}
\begin{tikzcd}
{} \arrow[r,shorten >= 12pt,-latex,start anchor=center,"\dim x" near start] \arrow[d,shorten >= 9pt,-latex,swap,start anchor=center,"\dim y" near start] &
    a  \arrow[rrrrrrrr,"p\at x"] \arrow[dd,equal] \arrow[ddrrrrrrrr, phantom, "\filler_{\dim y}((p \sq \refl_b)\at x)"] &&&&&&&& b \arrow[dd,equal] \\
 {} &   & \\
&   a \arrow[rrrrrrrr, swap,"(p \sq \refl_b)\at x"] &&&&&&&& b
\end{tikzcd}
\vspace{2mm}

\noindent The required identification follows by (double) dimension abstraction on the above square.
\end{proof}

We still need a few lemmas to show that the left unit property is true too, so we shall postpone it to the end of this section. 
One very useful proposition that can be proven at this point, however, is that the composition of any identification with its inverse equals the reflexivity element up to globular identification. But the proof is a little more involved than those of the preceding lemmas. So far we have only encountered one-extent Kan composition problems, which means that we have only considered open squares. Even when we were explicitly required to construct a two-dimensional identification (as in Lemmas \ref{lem:invu} to \ref{lem:ru}), we were able to found alternative ways to deal with the constructions without having to appeal to higher-dimensional Kan composition scenarios. 

It is now time to tackle truly higher-dimensional problems. From now on we will often work with more complex (two-extent) Kan composition scenarios. We start with the following lemma:

%%%%

\begin{lemma}[Right cancellation] \label{lem:rc}
For every degenerate $\dim x$-type $A$ and every $a \in A \dimsub{0}{x}$ and $b \in A \dimsub{1}{x}$ we have an identification
$$\mathsf{rc}_p \in \refl_a =_{\dim y.a =_{\dim x.A} a} p \sq p^{-1}$$
for any $p \in a =_{\dim x.A} b$.
\end{lemma}
\begin{proof}
We shall construct the required identification by a two-extent homogeneous Kan composition. In the one-dimensional case, it is enough to form an open square to perform a homogeneous Kan composition, while in the two-dimensional case we are required to form an open cube. In other words, we are expected to form an open $(\dim x, \dim y, \dim z)$-cube by finding one $(\dim x, \dim z)$-square (top), two $(\dim z, \dim y)$-squares (left and right) and two $(\dim x, \dim y)$-squares (back and front) whose faces all agree up to strict equality before we can obtain its lid: the composite $(\dim x, \dim z)$-square that forms the bottom face of the cube.

For this particular lemma this means that we must construct an open cube whose composite is an $(\dim x, \dim z)$-square with $\refl_a \at x$ and $p \sq p^{-1} \at x$ as $\dim x$-lines and $a$ in both degenerate $\dim z$-lines. Now consider the following open $(\dim x, \dim y, \dim z)$-cube (its composite is illustrated as the shaded face in the diagram below)

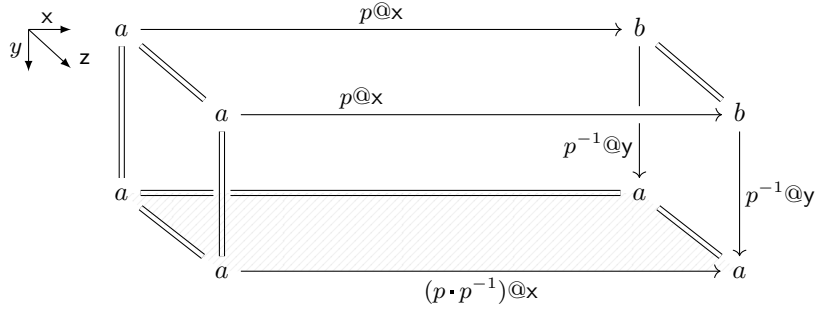
\vspace{2mm}
\begin{tikzcd}
[execute at end picture={
\foreach \Valor/\Nombre in
{
  tikz@f@9-3-2/a,tikz@f@9-4-3/b,tikz@f@9-3-9/c,tikz@f@9-4-10/d%
}
{
\coordinate (\Nombre) at (\Valor);
}
\fill[pattern=north east lines,pattern color=grey,opacity=0.3]
  (b) -- (a) -- (c) -- (d) -- cycle;
  }
]
{} \arrow[r,shorten >= 12pt,-latex,start anchor=center,"\dim x" near start] \arrow[dr,shorten >= 24pt,-latex,start anchor=center,"\dim z"] \arrow[d,shorten >= 15pt,-latex,swap,start anchor=center,"y" near start] & {}
%&%&
  a \arrow[rrrrrrr,"p \at x"] \arrow[dr,equal] \arrow[dd,equal] \arrow[drrrrrrrr, phantom] && &&&&&
                   b \arrow[dd,swap,"p^{-1} \at y" near end] \arrow[dr,equal] \\
{} & {} 
& a \arrow[rrrrrrr,crossing over,"p \at x" near start] &&    &&&&&    b \arrow[dd,"p^{-1} \at y"] \\
& a \arrow[rrrrrrr,equal] \arrow[dr,equal] &&    &&&&&   a \arrow[dr,equal] \\
& & a \arrow[rrrrrrr,swap,"(p \sq p^{-1}) \at x"] \arrow[uu,<-,crossing over,equal] &&  &&&&& a
\end{tikzcd}
\vspace{2mm}

%% Top:

\noindent whose top face is the $(\dim x, \dim z)$-square,

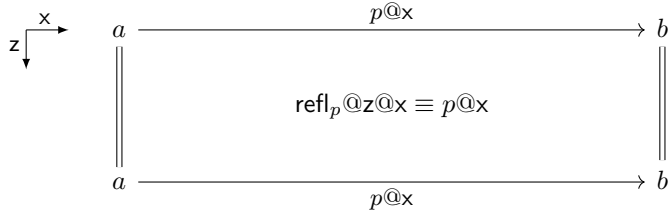
\vspace{2mm}
\begin{tikzcd}
{} \arrow[r,shorten >= 12pt,-latex,start anchor=center,"\dim x" near start] \arrow[d,shorten >= 9pt,-latex,swap,start anchor=center,"\dim z" near start] &
    a  \arrow[rrrrrrrr,"p\at x"] \arrow[dd,equal] \arrow[ddrrrrrrrr, phantom, "\refl_p \at z \at x \equiv p \at x"] &&&&&&&& b \arrow[dd,equal] \\
 {} &   & \\
&   a \arrow[rrrrrrrr,swap,"p\at x"] &&&&&&&& b
\end{tikzcd}
\vspace{2mm}

%% Right and Left:

\noindent where $\refl_p \in p =_{\dim y. a =_{\dim x.A} b} p$, left and right faces are respectively the $(\dim z, \dim y)$-squares 

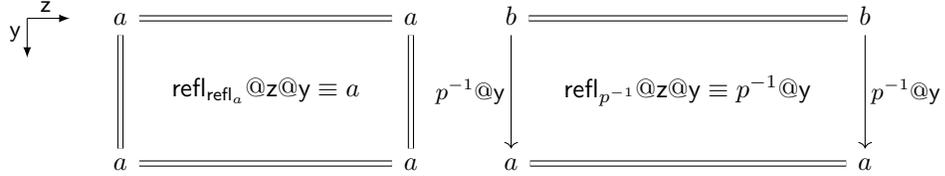
\vspace{2mm}
\begin{tikzcd}
{} \arrow[r,shorten >= 12pt,-latex,start anchor=center,"\dim z" near start] \arrow[d,shorten >= 9pt,-latex,swap,start anchor=center,"\dim y" near start] &
    a  \arrow[rrrr,equal] \arrow[dd,equal] \arrow[ddrrrr, phantom, "\refl_{\refl_a} \at z \at y \equiv a"] &&&& a \arrow[dd,equal] & b \arrow[rrrrr,equal] \arrow[dd,swap,"p^{-1} \at y"] \arrow[ddrrrrr, phantom,"\refl_{p^{-1}} \at z \at y \equiv p^{-1} \at y"] &&&&& b \arrow[dd,"p^{-1} \at  y"]\\
 {} &   & &&& {} &   & \\
&   a \arrow[rrrr,equal] &&&& a & a \arrow[rrrrr,equal] &&&&& a
\end{tikzcd}
\vspace{2mm}

%% Back and Front:

\noindent and back and front are respectively the $(\dim x, \dim y)$-squares

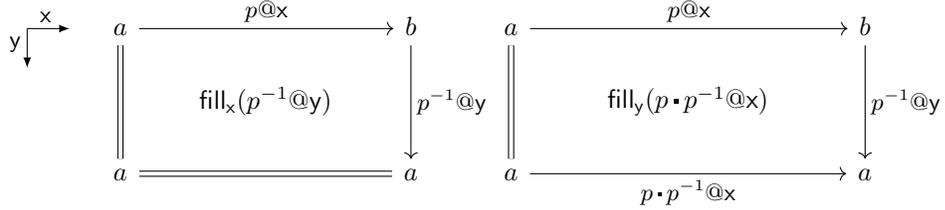
\vspace{2mm}
\begin{tikzcd}
{} \arrow[r,shorten >= 12pt,-latex,start anchor=center,"\dim x" near start] \arrow[d,shorten >= 9pt,-latex,swap,start anchor=center,"\dim y" near start] &
    a  \arrow[rrrr,"p \at x"] \arrow[dd,equal] \arrow[ddrrrr, phantom, "\filler_{\dim x}(p^{-1}\at y)"] &&&& b \arrow[dd,"p^{-1} \at  y"] & a \arrow[rrrrr,"p\at x"] \arrow[dd,equal] \arrow[ddrrrrr, phantom,"\filler_{\dim y}(p \sq p^{-1}\at x)"] &&&&& b \arrow[dd,"p^{-1} \at  y"]\\
 {} &   & &&& {} &   & \\
&   a \arrow[rrrr,equal] &&&& a & a \arrow[rrrrr, swap,"p \sq p^{-1}\at x"] &&&&& a
\end{tikzcd}
\vspace{2mm}

\noindent both of which are given by the fillers of the homogeneous Kan composition scenarios from Lemmas \ref{lem:inv} and \ref{lem:comp}, respectively. 

%% Bottom:

Now we note that the bottom $(\dim x, \dim z)$-square of the open cube described above is the homogeneous Kan composite
$$\hcom^{\lto{0}{1}}_{\dim A}(p\at x)({}^\dim{0}_\dim{x} \htri \dim y. a, \; {}^\dim{1}_\dim{x} \htri \dim y. p^{-1} \at y, \;
{}^\dim{0}_\dim{z} \htri \dim y.\filler_{\dim x}(p^{-1}\at y), \; {}^\dim{1}_\dim{z} \htri \dim y. \filler_{\dim y}((p \sq p^{-1})\at x)).$$

\end{proof}

%%%%

As we shall see in details next section, our inversion and composition operations have a very limited applicability. Recall that both are only well-defined for degenerate one-dimensional types, so we cannot, in general, invert an arbitrary identification 
$$\alpha \in p =_{\dim y.r \at y =_{\dim x.A} s \at y} q$$
because the $\dim y$-type $r \at y =_{\dim x.A} s \at y$ need not be degenerate (in fact, $r \at y$ and $s \at y$ need not be degenerate lines either). However, it seems natural to expect that $\alpha$ could somehow be `swapped' into an identification inhabiting the type
$$(q=_{\dim{y}.{r}^{-1} \at y =_{\dim x.A} {s}^{-1} \at y}p)$$
as long as $A$ is a degenerate $(\dim x, \dim y)$-type.

This is indeed the case, but requires another proof.

%%%%

\begin{lemma}[Square swap] \label{lem:swap}
For every degenerate $(\dim x, \dim y)$-type $A$ and every $p \in a =_{\dim x.A \dimsub{0}{y}} b$, $q \in c =_{\dim x.A \dimsub{1}{y}} d$, $r \in a =_{\dim y.A \dimsub{0}{x}} c$, $s \in b =_{\dim y.A \dimsub{1}{x}} d$, there exists an operation 
$$\mathsf{swap}_\alpha \in (p=_{\dim{y}.r \at y =_{\dim x.A} s \at y}q) \to (q=_{\dim{y}.{r}^{-1} \at y =_{\dim x.A} {s}^{-1} \at y}p)$$
where $a \in A\dimsub{0}{x}\dimsub{0}{y}$, $b \in A\dimsub{1}{x}\dimsub{0}{y}$, $c \in A\dimsub{0}{x}\dimsub{1}{y}$, $d \in A\dimsub{1}{x}\dimsub{1}{y}$.
\end{lemma}
\begin{proof}
The idea of the proof is to define a function that maps any $(\dim x, \dim y)$-square

\vspace{2mm}
\begin{tikzcd}
{} \arrow[r,shorten >= 12pt,-latex,start anchor=center,"\dim x" near start] \arrow[d,shorten >= 9pt,-latex,swap,start anchor=center,"\dim y" near start] &
    a  \arrow[rrrrrrrr,"p\at x"] \arrow[dd,"r \at y"] \arrow[ddrrrrrrrr, phantom, "\alpha \at y \at x"] &&&&&&&& b \arrow[dd,"s\at y"] \\
 {} &   & \\
&   c \arrow[rrrrrrrr, swap,"q\at x"] &&&&&&&& d
\end{tikzcd}
\vspace{2mm}

\noindent to a `swapped' $(\dim x, \dim y)$-square

\vspace{2mm}
\begin{tikzcd}
{} \arrow[r,shorten >= 12pt,-latex,start anchor=center,"\dim x" near start] \arrow[d,shorten >= 9pt,-latex,swap,start anchor=center,"\dim y" near start] &
    c  \arrow[rrrrrrrr,"q\at x"] \arrow[dd,"r^{-1} \at y"] \arrow[ddrrrrrrrr, phantom, "\mathsf{swap}_\alpha \at y \at x"] &&&&&&&& d \arrow[dd,"s^{-1} \at y"] \\
 {} &   & \\
&   a \arrow[rrrrrrrr, swap,"p\at x"] &&&&&&&& b
\end{tikzcd}
\vspace{2mm}

In order to obtain $\mathsf{swap}_\alpha \at y \at x$ we perform a two-extent homogeneous Kan composition on the open $(\dim x, \dim y, \dim z)$-cube formed by $p \at z$ at the top, $\alpha\at y \at z$ at the left, $p \at z$ at the right, $\filler_{\dim y}({r}^{-1}\at x)$ at the back and $\filler_{\dim y}({s}^{-1}\at x)$ at the front.

\vspace{2mm}
\begin{tikzcd}
[execute at end picture={
\foreach \Valor/\Nombre in
{
  tikz@f@15-3-2/a,tikz@f@15-4-3/b,tikz@f@15-3-9/c,tikz@f@15-4-10/d%
}
{
\coordinate (\Nombre) at (\Valor);
}
\fill[pattern=north east lines,pattern color=grey,opacity=0.3]
  (b) -- (a) -- (c) -- (d) -- cycle;
  }
]
{} \arrow[r,shorten >= 12pt,-latex,start anchor=center,"\dim x" near start] \arrow[dr,shorten >= 24pt,-latex,start anchor=center,"\dim z"] \arrow[d,shorten >= 15pt,-latex,swap,start anchor=center,"y" near start] & {}
%&%&
  a \arrow[rrrrrrr,equal] \arrow[dr,"p \at z"] \arrow[dd,swap,"r \at y"] \arrow[drrrrrrrr, phantom] && &&&&&
                   a \arrow[dd,equal] \arrow[dr,"p \at z"] \\
{} & {} 
& b \arrow[rrrrrrr,crossing over,equal] &&    &&&&&    b \arrow[dd,equal] \\
& c \arrow[rrrrrrr,"{r}^{-1} \at x" near end] \arrow[dr,swap,"q \at z"] && &&&&& a \arrow[dr,"p \at z" near start] \\
& & d \arrow[rrrrrrr,swap,"{s}^{-1} \at x"] \arrow[uu,<-,crossing over,swap,"s \at y" near end] &&  &&&&& b
\end{tikzcd}
\vspace{2mm}

We thus define $\mathsf{swap}_\alpha$ by
$$\dangle{x}\dangle{z} \hcom^{\lto{0}{1}}_{\dim A}(p\at z)({}^\dim{0}_\dim{x} \htri \dim y. \alpha\at y \at z, \; {}^\dim{1}_\dim{x} \htri \dim y. p \at z, \;
{}^\dim{0}_\dim{z} \htri \dim y.\filler_{\dim y}({r}^{-1}\at x), \; {}^\dim{1}_\dim{z} \htri \dim y. \filler_{\dim y}({s}^{-1}\at x)).$$

\end{proof}

%%%%

One application of the square swap lemma is in the following proof that a double inverted identification equals the original identification up to globular identification (to put it another way, double inversion is essentially redundant).

%%%%

\begin{lemma}[Inversability] \label{lem:inversab}
For every degenerate $\dim x$-type $A$ and every $a \in A \dimsub{0}{x}$ and $b \in A \dimsub{1}{x}$, we have an identification
$$\mathsf{inv}_p \in p =_{\dim y. a =_{\dim x.A} b} (p^{-1})^{-1}$$
for any $p \in a =_{\dim x.A} b$.
\end{lemma}
\begin{proof}
By homogeneous Kan composition.
It suffices to find an $(\dim x, \dim z)$-square (which, for future reference, we shall call $R$) for the top face of the open cube

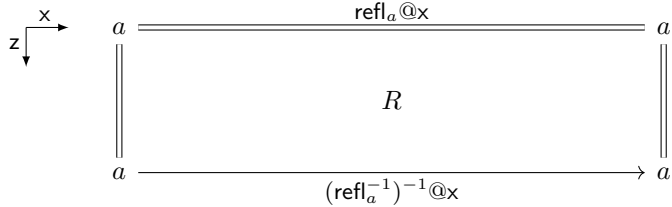
\vspace{2mm}
\begin{tikzcd}
{} \arrow[r,shorten >= 12pt,-latex,start anchor=center,"\dim x" near start] \arrow[d,shorten >= 9pt,-latex,swap,start anchor=center,"\dim z" near start] &
    a  \arrow[rrrrrrrr,equal,"\refl_a \at x"] \arrow[dd,equal] \arrow[ddrrrrrrrr, phantom, "R"] &&&&&&&& a \arrow[dd,equal] \\
 {} &   & \\
&   a \arrow[rrrrrrrr,swap,"(\refl_a^{-1})^{-1}\at x"] &&&&&&&& a
\end{tikzcd}
\vspace{2mm}

\noindent and two $(\dim x, \dim y)$-squares (which we shall call $X$ and $Y$) for the back and front faces of the open cube

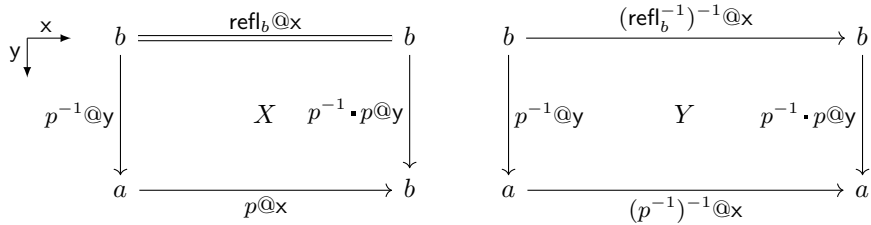
\vspace{2mm}
\begin{tikzcd}
{} \arrow[r,shorten >= 12pt,-latex,start anchor=center,"\dim x" near start] \arrow[d,shorten >= 9pt,-latex,swap,start anchor=center,"\dim y" near start] &
    b  \arrow[rrrr,equal,"\refl_b \at x"] \arrow[dd,swap,"p^{-1} \at  y"] \arrow[ddrrrr, phantom, "X"] &&&& b \arrow[dd,swap,"p^{-1} \sq p \at  y"] & b \arrow[rrrrr,"(\refl_b^{-1})^{-1} \at x"] \arrow[dd,"p^{-1} \at  y"] \arrow[ddrrrrr, phantom,"Y"] &&&&& b \arrow[dd,swap,"p^{-1} \sq p \at  y"]\\
 {} &   & &&& {} &   & \\
&   a \arrow[rrrr, swap,"p\at x"] &&&& b & a \arrow[rrrrr, swap,"(p^{-1})^{-1}\at x"] &&&&& a
\end{tikzcd}
\vspace{2mm}

\noindent (the remaining sides of the open cube will be composed of degenerate squares.)

Note that $R$ basically states that the lemma is true when $p$ is $\refl_a$, that is to say, it represents the identification
$$\refl_a =_{\dim y.a =_{\dim x.A} a} (\refl_a^{-1})^{-1}.$$
We construct $R$ by homogeneous Kan composition on an the open cube formed by the inversion unit square from Lemma \ref{lem:inv} at the right, the filler of the Kan composite $(\refl_a^{-1})^{-1}$ from Lemma \ref{lem:inv} at the front (and degenerate squares at the remaining faces), as can be seen in the diagram

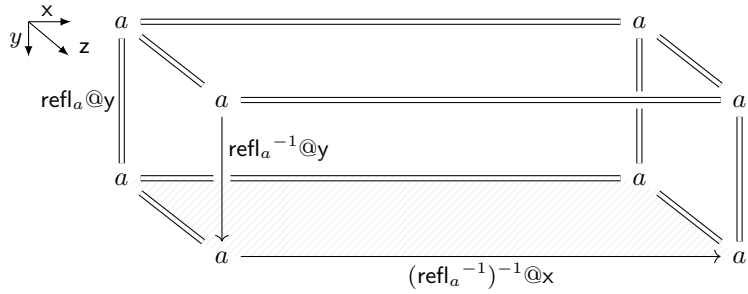
\vspace{2mm}
\begin{tikzcd}
[execute at end picture={
\foreach \Valor/\Nombre in
{
  tikz@f@18-3-2/a,tikz@f@18-4-3/b,tikz@f@18-3-9/c,tikz@f@18-4-10/d%
}
{
\coordinate (\Nombre) at (\Valor);
}
\fill[pattern=north east lines,pattern color=grey,opacity=0.3]
  (b) -- (a) -- (c) -- (d) -- cycle;
  }
]
{} \arrow[r,shorten >= 12pt,-latex,start anchor=center,"\dim x" near start] \arrow[dr,shorten >= 24pt,-latex,start anchor=center,"\dim z"] \arrow[d,shorten >= 15pt,-latex,swap,start anchor=center,"y" near start] & {}
%&%&
  a \arrow[rrrrrrr,equal] \arrow[dr,equal] \arrow[dd,swap,equal,"\refl_a\at y"]  && &&&&&
                   a \arrow[dd,equal] \arrow[dr,equal] \\
{} & {} 
& a \arrow[rrrrrrr,crossing over,equal] &&    &&&&&    a \arrow[dd,equal] \\
& a \arrow[rrrrrrr,equal] \arrow[dr,equal] &&    &&&&&   a \arrow[dr,equal] \\
& & a \arrow[rrrrrrr,swap,"({\refl_a}^{-1})^{-1} \at x"] \arrow[uu,<-,crossing over,swap,"{\refl_a}^{-1} \at y" near end] &&  &&&&& a
\end{tikzcd}
\vspace{2mm}

\noindent More precisely, we let

$$R :\equiv \hcom^{\lto{0}{1}}_{\dim A}(a)({}^\dim{0}_\dim{x} \htri \dim y. \mathsf{iu}(a)\at z \at y, \; {}^\dim{1}_\dim{x} \htri \dim y. a, \;
{}^\dim{0}_\dim{z} \htri \dim y.a, \; {}^\dim{1}_\dim{z} \htri \dim y. \filler_{\dim y}((\refl_a^{-1})^{-1}\at x)).$$

The construction of $X$ is an immediate consequence of Lemma \ref{lem:comp}: we define it as the $\dim x$-filler of the Kan composite $(p^{-1} \sq p) \at y$ (regarded as a $\dim y$-line)

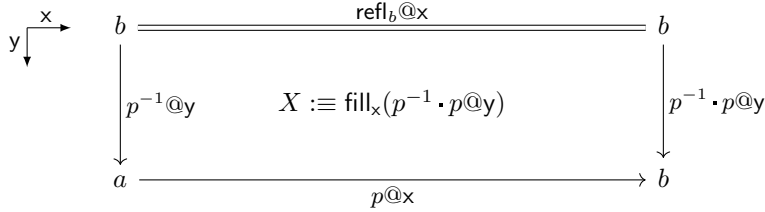
\vspace{2mm}
\begin{tikzcd}
{} \arrow[r,shorten >= 12pt,-latex,start anchor=center,"\dim x" near start] \arrow[d,shorten >= 9pt,-latex,swap,start anchor=center,"\dim y" near start] &
    b  \arrow[rrrrrrrr,equal,"\refl_b\at x"] \arrow[dd,"p^{-1}\at y"] \arrow[ddrrrrrrrr, phantom, "X :\equiv \filler_{\dim x}(p^{-1} \sq p \at y)"] &&&&&&&& b \arrow[dd,"p^{-1} \sq p \at y"] \\
 {} &   & \\
&   a \arrow[rrrrrrrr, swap,"p\at x"] &&&&&&&& b
\end{tikzcd}
\vspace{2mm}

\noindent The key to the construction of $Y$ is to observe that it is very similar to $X$, except that the $\dim y$-lines that forms the left and right faces of $X$ are not double inverted like the left and right faces of $Y$ are. Square swapping (Lemma \ref{lem:swap}) provides a method of double inverting the left and right faces of $X$ without altering its top and bottom faces. This can be done in two simple steps. 

First we obtain the $(\dim x, \dim y)$-square

\vspace{2mm}
\begin{tikzcd}
{} \arrow[r,shorten >= 12pt,-latex,start anchor=center,"\dim x" near start] \arrow[d,shorten >= 9pt,-latex,swap,start anchor=center,"\dim y" near start] &
    b  \arrow[rrrrrrrr,"\refl_b^{-1}\at x"] \arrow[dd,"p^{-1} \sq p \at y"] \arrow[ddrrrrrrrr, phantom, "\mathsf{swap}(\dangle{x}\dangle{y} X) \at x \at y"] &&&&&&&& b \arrow[dd,"p^{-1} \at y"] \\
 {} &   & \\
&   b \arrow[rrrrrrrr, swap,"p^{-1}\at x"] &&&&&&&& a
\end{tikzcd}
\vspace{2mm}

\noindent and then swap it again into

\vspace{2mm}
\begin{tikzcd}
{} \arrow[r,shorten >= 12pt,-latex,start anchor=center,"\dim x" near start] \arrow[d,shorten >= 9pt,-latex,swap,start anchor=center,"\dim y" near start] &
    b  \arrow[rrrrrrrr,"(\refl_b^{-1})^{-1}\at x"] \arrow[dd,"p^{-1}\at y"] \arrow[ddrrrrrrrr, phantom, "\mathsf{swap}(\mathsf{swap}(\dangle{x}\dangle{y} X)) \at x \at y"] &&&&&&&& b \arrow[dd,"p^{-1} \sq p \at y"] \\
 {} &   & \\
&   a \arrow[rrrrrrrr, swap,"(p^{-1})^{-1}\at x"] &&&&&&&& b
\end{tikzcd}
\vspace{2mm}

\noindent to obtain the required $(\dim x, \dim y)$-square $Y$.

Now that we have $R$, $X$ and $Y$, we can define $\mathsf{inv}_p$ by
$$\dangle{z}\dangle{x} \hcom^{\lto{0}{1}}_{\dim A}(R)({}^\dim{0}_\dim{x} \htri \dim y. p^{-1} \at y, \; {}^\dim{1}_\dim{x} \htri \dim y. (p^{-1} \sq p) \at y, \;
{}^\dim{0}_\dim{z} \htri \dim y.X, \; {}^\dim{1}_\dim{z} \htri \dim y.Y)$$
which basically represents the homogeneous Kan composite of the open cube

\vspace{2mm}
\begin{tikzcd}
[execute at end picture={
\foreach \Valor/\Nombre in
{
  tikz@f@22-3-2/a,tikz@f@22-4-3/b,tikz@f@22-3-9/c,tikz@f@22-4-10/d%
}
{
\coordinate (\Nombre) at (\Valor);
}
\fill[pattern=north east lines,pattern color=grey,opacity=0.3]
  (b) -- (a) -- (c) -- (d) -- cycle;
  }
]
{} \arrow[r,shorten >= 12pt,-latex,start anchor=center,"\dim x" near start] \arrow[dr,shorten >= 24pt,-latex,start anchor=center,"\dim z"] \arrow[d,shorten >= 15pt,-latex,swap,start anchor=center,"y" near start] & {}
%&%&
  b \arrow[rrrrrrr,equal] \arrow[dr,equal] \arrow[dd,swap,"p^{-1} \at y"] \arrow[drrrrrrrr, phantom] && &&&&&
                   b \arrow[dd,swap,"p^{-1} \sq p \at y" near start] \arrow[dr,equal] \\
{} & {} 
& b \arrow[rrrrrrr,crossing over,"{{\refl_b}^{-1}}^{-1} \at x" near start] &&    &&&&&    b \arrow[dd,"p^{-1} \sq p \at y"] \\
& a \arrow[rrrrrrr,"p \at x" near end] \arrow[dr,swap,equal] &&    &&&&&   b \arrow[dr,swap,equal] \\
& & a \arrow[rrrrrrr,swap,"(p^{-1})^{-1} \at x"] \arrow[uu,<-,crossing over,swap,"p^{-1} \at y" near end] &&  &&&&& b
\end{tikzcd}
\vspace{2mm}

\noindent formed by $R$ at the top, $p^{-1} \at y$ and $(p^{-1} \sq p) \at y$ respectively at the left and right and $X$ and $Y$ respectively at the back and front.

\end{proof}

%%%%

We hope that the reader is starting to get a feel for proofs by Kan composition and the interplay between two-dimensional identifications and squares at this point. 
Next we want to show that that left cancellation property holds as well (whose right counterpart is Lemma \ref{lem:rc}), but for this we will need the following lemma.

\begin{lemma}[Opposite identification] \label{lem:op}
For every degenerate $\dim x$-type $A$ and every $a \in A \dimsub{0}{x}$ and $b \in A \dimsub{1}{x}$, we have identifications
\begin{enumerate}[(i)]
	\item $\mathsf{op1}_{p} \in p =_{\dim y.a =_{\dim x.A} p^{-1} \at y} a$
	\item $\mathsf{op2}_{p} \in p^{-1} =_{\dim y.b =_{\dim x.A} p \at y} b$
\end{enumerate}
for any $p \in a =_{\dim x.A} b$.
\end{lemma}
\begin{proof} Both proofs use homogeneous Kan composition.
\begin{enumerate}[(i)]
	\item By Lemma \ref{lem:comp} we have an $(\dim x, \dim z)$-square $\filler_{\dim z}((p \sq p^{-1}) \at x)$ formed from the filler of the composition of $p$ and $p^{-1}$. By applying Lemma \ref{lem:rc} to $p$ we obtain an $(\dim x, \dim y)$-square $\mathsf{rc}_{p} \at y \at x$. We now construct an open $(\dim x, \dim y, \dim z)$ cube as follows: we put $\filler_{\dim z}((p \sq p^{-1}) \at x)$ at the top $(\dim x, \dim z)$-square, $a$ and $p^{-1} \at z$ at the left and right $(\dim z, \dim y)$-squares, respectively, and $p \at x$ and the square formed by inversion of $\mathsf{rc}_{p}$, that is, $(\mathsf{rc}_{p})^{-1} \at y \at x$ at the back and front $(\dim x, \dim y)$-squares, respectively.

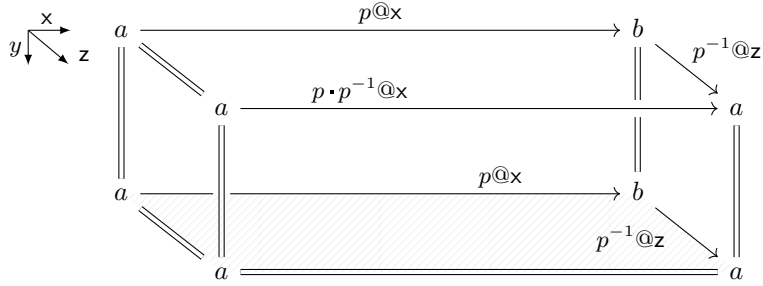
\vspace{2mm}
\begin{tikzcd}
[execute at end picture={
\foreach \Valor/\Nombre in
{
  tikz@f@23-3-2/a,tikz@f@23-4-3/b,tikz@f@23-3-9/c,tikz@f@23-4-10/d%
}
{
\coordinate (\Nombre) at (\Valor);
}
\fill[pattern=north east lines,pattern color=grey,opacity=0.3]
  (b) -- (a) -- (c) -- (d) -- cycle;
  }
]
{} \arrow[r,shorten >= 12pt,-latex,start anchor=center,"\dim x" near start] \arrow[dr,shorten >= 24pt,-latex,start anchor=center,"\dim z"] \arrow[d,shorten >= 15pt,-latex,swap,start anchor=center,"y" near start] & {}
%&%&
  a \arrow[rrrrrrr,"p \at x"] \arrow[dr,equal] \arrow[dd,equal] \arrow[drrrrrrrr, phantom] && &&&&&
                   b \arrow[dd,equal] \arrow[dr,"p^{-1} \at z"] \\
{} & {} 
& a \arrow[rrrrrrr,crossing over,"p \sq p^{-1} \at x" near start] &&    &&&&&    a \arrow[dd,equal] \\
& a \arrow[rrrrrrr,"p \at x" near end] \arrow[dr,equal] &&    &&&&&   b \arrow[dr,swap,"p^{-1} \at z" near start] \\
& & a \arrow[rrrrrrr,equal] \arrow[uu,<-,crossing over,equal] &&  &&&&& a
\end{tikzcd}
\vspace{2mm}

Thus, we define $\mathsf{op1}_p$ by

$$\dangle{z}\dangle{x} \hcom^{\lto{0}{1}}_{\dim A}(\filler_{\dim z}((p \sq p^{-1}) \at x))({}^\dim{0}_\dim{x} \htri \dim y. a, \; {}^\dim{1}_\dim{x} \htri \dim y. p^{-1} \at z, \;
{}^\dim{0}_\dim{z} \htri \dim y. p \at x, \; {}^\dim{1}_\dim{z} \htri \dim y.(\mathsf{rc}_{p})^{-1} \at y \at x).$$

	\item This may be verified by a similar argument as above.
\end{enumerate}
\end{proof}

We are now able to prove the left cancellation property, which states that the composition of an identification with its inverse equals the reflexivity element up to globular identification.

%%%%%%%%%%%%%%%%%%%%%%%%%%%%

\begin{lemma}[Left cancellation] \label{lem:lc}
For every degenerate $\dim x$-type $A$ and every $a \in A \dimsub{0}{x}$ and $b \in A \dimsub{1}{x}$ we have an identification
$$\mathsf{lc}_p \in \refl_b =_{\dim y. b =_{\dim x.A} b} p^{-1} \sq p$$
for any $p \in a =_{\dim x.A} b$.
\end{lemma}
\begin{proof} 
By homogeneous Kan composition on the following open cube

\vspace{2mm}
\begin{tikzcd}
[execute at end picture={
\foreach \Valor/\Nombre in
{
  tikz@f@24-3-2/a,tikz@f@24-4-3/b,tikz@f@24-3-9/c,tikz@f@24-4-10/d%
}
{
\coordinate (\Nombre) at (\Valor);
}
\fill[pattern=north east lines,pattern color=grey,opacity=0.3]
  (b) -- (a) -- (c) -- (d) -- cycle;
  }
]
{} \arrow[r,shorten >= 12pt,-latex,start anchor=center,"\dim x" near start] \arrow[dr,shorten >= 24pt,-latex,start anchor=center,"\dim z"] \arrow[d,shorten >= 15pt,-latex,swap,start anchor=center,"y" near start] & {}
%&%&
  b \arrow[rrrrrrr,"p^{-1} \at x"] \arrow[dr,equal] \arrow[dd,equal] \arrow[drrrrrrrr, phantom] && &&&&&
                   a \arrow[dd,swap,"(p^{-1})^{-1} \at y" near start] \arrow[dr,equal] \\
{} & {} 
& b \arrow[rrrrrrr,crossing over,"p^{-1} \at x" near start] &&    &&&&&    a \arrow[dd,"p \at y"] \\
& b \arrow[rrrrrrr,equal,"\refl_b \at x" near end] \arrow[dr,equal] &&    &&&&&   b \arrow[dr,equal] \\
& & b \arrow[rrrrrrr,swap,"p^{-1} \sq p \at x"] \arrow[uu,<-,crossing over,equal] &&  &&&&& b
\end{tikzcd}
\vspace{2mm}

\noindent So our desired identification is 
\begin{multline*}
\mathsf{lc}_p :\equiv \dangle{z}\dangle{x} \hcom^{\lto{0}{1}}_{\dim A}(p^{-1} \at x)({}^\dim{0}_\dim{x} \htri \dim y. b, \; {}^\dim{1}_\dim{x} \htri \dim y. (\mathsf{inv}_p)^{-1}\at z \at y, \; \\ 
{}^\dim{0}_\dim{z} \htri \dim y. \mathsf{op1}_{p^{-1}} \at y \at x, \; {}^\dim{1}_\dim{z} \htri \dim y.\filler_{\dim y}(p^{-1} \sq p\at x)).
\end{multline*}
where the $(\dim x, \dim y)$-square $\mathsf{op1}_{p^{-1}} \at y \at x$ is obtained by Lemma \ref{lem:op} (i).

\end{proof}

We are now finally ready to show that the left counterpart of the unit property from Lemma \ref{lem:ru} is the case, or, put differently, that the reflexivity element is a left unit for composition.

%%%%%%%%%%%%%%%%%%%%%%%%

\begin{lemma}[Left unit] \label{lem:lu}
For every degenerate $\dim x$-type $A$ and every $a \in A \dimsub{0}{x}$ and $b \in A \dimsub{1}{x}$ we have an identification
$$\mathsf{lu}_p \in p =_{\dim y. a =_{\dim x.A} b} \refl_a \sq p$$
for any $p \in a =_{\dim x.A} b$.
\end{lemma}
\begin{proof}
As before, the proof follows by Kan composition,

\vspace{2mm}
\begin{tikzcd}
[execute at end picture={
\foreach \Valor/\Nombre in
{
  tikz@f@25-3-2/a,tikz@f@25-4-3/b,tikz@f@25-3-9/c,tikz@f@25-4-10/d%
}
{
\coordinate (\Nombre) at (\Valor);
}
\fill[pattern=north east lines,pattern color=grey,opacity=0.3]
  (b) -- (a) -- (c) -- (d) -- cycle;
  }
]
{} \arrow[r,shorten >= 12pt,-latex,start anchor=center,"\dim x" near start] \arrow[dr,shorten >= 24pt,-latex,start anchor=center,"\dim z"] \arrow[d,shorten >= 15pt,-latex,swap,start anchor=center,"y" near start] & {}
%&%&
  a \arrow[rrrrrrr,"p \at x"] \arrow[dr,equal] \arrow[dd,swap,equal] \arrow[drrrrrrrr, phantom] && &&&&&
                   b \arrow[dd,equal,near end] \arrow[dr,"p^{-1} \at z"] \\
{} & {} 
& a \arrow[rrrrrrr,crossing over,equal,near start] &&    &&&&&    a \arrow[dd,"p \at y"] \\
& a \arrow[rrrrrrr,"p \at x" near end] \arrow[dr,swap,equal] &&    &&&&&   b \arrow[dr,equal] \\
& & a \arrow[rrrrrrr,swap,"\refl_a \sq p \at x"] \arrow[uu,<-,crossing over,equal,near end] &&  &&&&& b
\end{tikzcd}
\vspace{2mm}

and here we have
\begin{multline*}
\mathsf{lu}_p :\equiv \dangle{z}\dangle{x} \hcom^{\lto{0}{1}}_{\dim A}(\filler_{\dim x}(p^{-1} \at z))({}^\dim{0}_\dim{x} \htri \dim y. a, \; {}^\dim{1}_\dim{x} \htri \dim y. \mathsf{op2}_p \at y \at z, \; \\ 
{}^\dim{0}_\dim{z} \htri \dim y. p \at x, \; {}^\dim{1}_\dim{z} \htri \dim y.\filler_{\dim y}((\refl_a \sq p)\at x)),
\end{multline*}
where the $(\dim x, \dim z)$-square $\mathsf{op2}_p \at y \at z$ is obtained by Lemma \ref{lem:op} (ii).

\end{proof}

The curious reader may wonder why our proof of the left unit property from Lemma \ref{lem:lu} is significantly harder than the right unit one (Lemma \ref{lem:ru}). Why the property is so much simpler to demonstrate in the right? 
If we look attentively at the filler of, say, $p \sq q$ from Lemma \ref{lem:comp}, 

\vspace{2mm}
\begin{tikzcd}
{} \arrow[r,shorten >= 12pt,-latex,start anchor=center,"\dim x" near start] \arrow[d,shorten >= 9pt,-latex,swap,start anchor=center,"\dim y" near start] & a
       \arrow[rrrrrrrr,"p\at x"] \arrow[dd,equal] \arrow[ddrrrrrrrr, phantom, "\filler_{\dim y}(p \sq q \at x)"] &&&&&&&& b \arrow[dd,"q \at y"] \\
 {} &   & \\
&   a \arrow[rrrrrrrr, swap,"p \sq q \at x"] &&&&&&&& c
\end{tikzcd}
\vspace{2mm}

\noindent we can see that it forms a simultaneous identification: an identification of $p \at x$ and $p \sq q \at x$ as $\dim x$-lines modulo an identification of $q \at y$ and $\refl_a \at y$ as $\dim y$-lines. Consequently, if we set $q :\equiv \refl_a$, then we immediately have a globular identification of $p$ and $p \sq q$. We can thus compare our composition operation with one defined in conventional homotopy type theory \cite{hottbook} by \emph{path induction} on the second argument $q$ \cite{hottbook}[\S 1.12.1], since we let $p \sq q$ be $p$ just in case $q$ is $\refl_a$. The same remark applies to inversion as well: in this case $p^{-1}$ is related to an inversion operation defined by path induction on $p$ by letting $p^{-1}$ be $\refl_a$ just in case $p$ is $\refl_a$ (in fact, we will see in \Cref{moreabout} that path induction is just a particular case of the Kan conditions.)

Last but not least, we want to show that composition of identifications is associative up to globular identification. For this we use the following lemma, which basically states that any two squares with strictly equal top, right and left faces must have the same bottom up to globular identification.

%%%%

\begin{lemma}[The three-out-of-four bottom identification] \label{lem:3o4}
For every degenerate $(\dim x, \dim y)$-type $A$ and every $\alpha \in p=_{\dim{y}.r \at y =_{\dim x.A} s \at y}q$, $\beta \in p=_{\dim{y}.r \at y =_{\dim x.A} s \at y}q'$ we have an identification
$$\mathsf{bi}_{\alpha,\beta} \in q =_{\dim y. c =_{\dim x.A} d} q'$$
where $p \in a =_{\dim x.A \dimsub{0}{y}} b$, $q, q' \in c =_{\dim x.A \dimsub{1}{y}} d$, $r \in a =_{\dim y.A \dimsub{0}{x}} c$, $s \in b =_{\dim y.A \dimsub{1}{x}} d$ and $a \in A\dimsub{0}{x}\dimsub{0}{y}$, $b \in A\dimsub{1}{x}\dimsub{0}{y}$, $c \in A\dimsub{0}{x}\dimsub{1}{y}$, $d \in A\dimsub{1}{x}\dimsub{1}{y}$.
\end{lemma}
\begin{proof}
By assumption, we have two $(\dim x, \dim y)$-squares $\alpha \at y \at x$ and $\beta \at y \at x$ with strictly equal top faces $p \at x$, left faces $r \at y$ and right faces $s \at y$. We want to show that the bottom faces of $\alpha \at y \at x$ and $\beta \at y \at x$, which are respectively $q \at x$ and $q' \at x$, are equal up to globular identification.

We do this by finding an $(\dim x, \dim y)$-square 

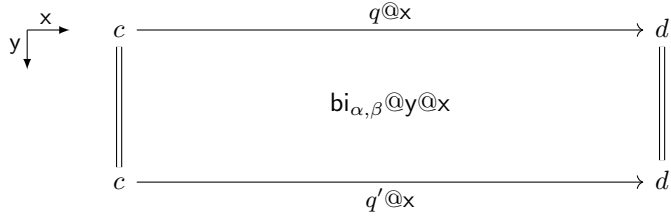
\vspace{2mm}
\begin{tikzcd}
{} \arrow[r,shorten >= 12pt,-latex,start anchor=center,"\dim x" near start] \arrow[d,shorten >= 9pt,-latex,swap,start anchor=center,"\dim y" near start] & c
       \arrow[rrrrrrrr,"q\at x"] \arrow[dd,equal] \arrow[ddrrrrrrrr, phantom, "\mathsf{bi}_{\alpha,\beta} \at y \at x"] &&&&&&&& d \arrow[dd,equal] \\
 {} &   & \\
&   c \arrow[rrrrrrrr, swap,"q' \at x"] &&&&&&&& d
\end{tikzcd}
\vspace{2mm}

\noindent by homogeneous Kan composition on the following open $(\dim x, \dim y, \dim z)$-cube 

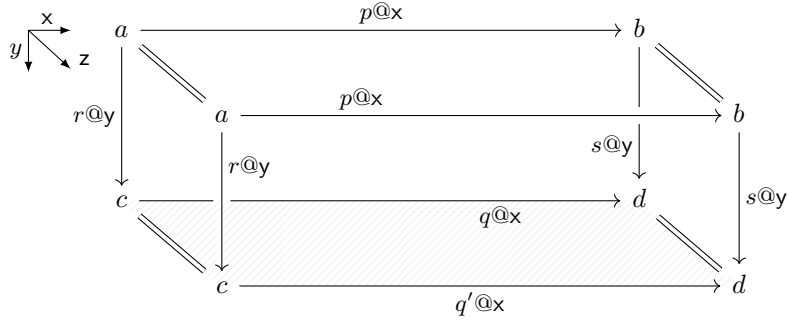
\vspace{2mm}
\begin{tikzcd}
[execute at end picture={
\foreach \Valor/\Nombre in
{
  tikz@f@28-3-2/a,tikz@f@28-4-3/b,tikz@f@28-3-9/c,tikz@f@28-4-10/d%
}
{
\coordinate (\Nombre) at (\Valor);
}
\fill[pattern=north east lines,pattern color=grey,opacity=0.3]
  (b) -- (a) -- (c) -- (d) -- cycle;
  }
]
{} \arrow[r,shorten >= 12pt,-latex,start anchor=center,"\dim x" near start] \arrow[dr,shorten >= 24pt,-latex,start anchor=center,"\dim z"] \arrow[d,shorten >= 15pt,-latex,swap,start anchor=center,"y" near start] & {}
%&%&
  a \arrow[rrrrrrr,"p \at x"] \arrow[dr,equal] \arrow[dd,swap,"r \at y"] \arrow[drrrrrrrr, phantom] && &&&&&
                   b \arrow[dd,swap,"s \at y" near end] \arrow[dr,equal] \\
{} & {} 
& a \arrow[rrrrrrr,crossing over,"p \at x" near start] &&    &&&&&    b \arrow[dd,"s \at y"] \\
& c \arrow[rrrrrrr,swap,"q \at x" near end] \arrow[dr,equal] &&    &&&&&  d \arrow[dr,equal] \\
& & c \arrow[rrrrrrr,swap,"q' \at x"] \arrow[uu,<-,crossing over,swap,"r \at y" near end] &&  &&&&& d
\end{tikzcd}
\vspace{2mm}

\noindent This open cube is formed by $p \at x$ as the top $(\dim x, \dim z)$-square, $r$ and $s$ as respectively the left and right $(\dim z, \dim y)$-squares and $\alpha \at y \at x$ and $\beta \at y \at x$ as respectively the back and front $(\dim x, \dim y)$-squares.

\end{proof}

%%%%

Now we can show that associativity holds up to globular identification:

%%%%

\begin{lemma}[Associativity] \label{lem:assoc}
For every degenerate $\dim x$-type $A$ and every $a \in A\dimsub{0}{x}$, $b \in A\dimsub{1}{x}$, $c \in A\dimsub{1}{x}$, $d \in A\dimsub{1}{x}$, we have an identification
$$\mathsf{assoc}_{p,q,r} \in (p \sq q) \sq r =_{\dim y. a =_{\dim x.A} d} p \sq (q \sq r) $$
for any $p \in a =_{\dim x.A} b$, $q \in b =_{\dim x.A} c$, $r \in c =_{\dim y.A} d$.
\end{lemma}
\begin{proof}
By routine diagram chasing. Homogeneous Kan composition ensures the existence of the $(\dim x, \dim y, \dim z)$-cube

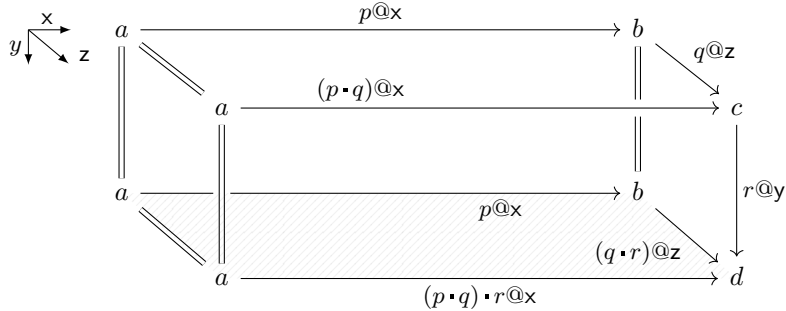
\vspace{2mm}
\begin{tikzcd}
[execute at end picture={
\foreach \Valor/\Nombre in
{
  tikz@f@29-3-2/a,tikz@f@29-4-3/b,tikz@f@29-3-9/c,tikz@f@29-4-10/d%
}
{
\coordinate (\Nombre) at (\Valor);
}
\fill[pattern=north east lines,pattern color=grey,opacity=0.3]
  (b) -- (a) -- (c) -- (d) -- cycle;
  }
]
{} \arrow[r,shorten >= 12pt,-latex,start anchor=center,"\dim x" near start] \arrow[dr,shorten >= 24pt,-latex,start anchor=center,"\dim z"] \arrow[d,shorten >= 15pt,-latex,swap,start anchor=center,"y" near start] & {}
%&%&
  a \arrow[rrrrrrr,"p \at x"] \arrow[dr,equal] \arrow[dd,equal] \arrow[drrrrrrrr, phantom] && &&&&&
                   b \arrow[dd,equal] \arrow[dr,"q \at z"] \\
{} & {} 
& a \arrow[rrrrrrr,crossing over,"(p \sq q) \at x" near start] &&    &&&&&    c \arrow[dd,"r \at y"] \\
& a \arrow[rrrrrrr,swap,"p \at x" near end] \arrow[dr,equal] &&    &&&&&   b \arrow[dr,swap,"(q \sq r) \at z"] \\
& & a \arrow[rrrrrrr,swap,"(p \sq q) \sq r \at x"] \arrow[uu,<-,crossing over,equal] &&  &&&&& d
\end{tikzcd}
\vspace{2mm}

\noindent but then we have two $(x,z)$-squares with strictly equal top, right and left faces 

\vspace{2mm}
\begin{tikzcd}
{} \arrow[r,shorten >= 12pt,-latex,start anchor=center,"\dim x" near start] \arrow[d,shorten >= 9pt,-latex,swap,start anchor=center,"\dim z" near start] &
    a  \arrow[rrrr,"p \at x"] \arrow[dd,equal] \arrow[ddrrrr, phantom] &&&& b \arrow[dd,swap,"(q \sq r) \at  z"] & a \arrow[rrrr,"p \at x"] \arrow[dd,equal] \arrow[ddrrrr, phantom,"\filler_{\dim z} ((p \sq (q \sq r)) \at x)"] &&&& b \arrow[dd,"(q \sq r) \at  z"]\\
 {} &   & &&& {} &   & \\
&   a \arrow[rrrr, swap,"(p \sq q) \sq r \at x"] &&&& d & a \arrow[rrrr, swap,"p \sq (q \sq r) \at x"] &&&& d
\end{tikzcd}
\vspace{2mm}

\noindent and, by Lemma \ref{lem:3o4}, they must have identical bottom faces. Thus, we have a square

\vspace{2mm}
\begin{tikzcd}
{} \arrow[r,shorten >= 12pt,-latex,start anchor=center,"\dim x" near start] \arrow[d,shorten >= 9pt,-latex,swap,start anchor=center,"\dim y" near start] & a
       \arrow[rrrrrrrr,"(p \sq q) \sq r \at x"] \arrow[dd,equal] \arrow[ddrrrrrrrr, phantom, "\mathsf{assoc}_{p,q,r} \at y \at x"] &&&&&&&& d \arrow[dd,equal] \\
 {} &   & \\
&   a \arrow[rrrrrrrr, swap,"p \sq (q \sq r) \at x"] &&&&&&&& d
\end{tikzcd}
\vspace{2mm}

\end{proof}

%%%%

%%

%%%%%%%%%%%%%%%%%%%%%%%%%%%%%%%%
\subsection{Heterogeneous groupoid operations} \label{hetgroupoid}
%%%%%%%%%%%%%%%%%%%%%%%%%%%%%%%%

A careful reader will probably wonder why we insisted in describing both the inversion function from Lemma \ref{lem:inv} and the composition function from Lemma \ref{lem:comp} as preliminary (but not definitive) definitions. This is because they share a fundamental limitation: they can only be applied to degenerate types, types that do not depend on the dimension name which is being abstracted in the identification type in consideration. We expressed this limitation explicitly by confining the applicability of our propositions to degenerate $\dim x$-types whenever we were dealing with an identification type $a =_{\dim x.A} b$, for this condition guarantees that $a =_{\dim x.A} b$ is a homogeneous identification type (which means that the abstracted dimension name $\dim x$ does not occur in $A$).

Let us consider the limitations of our preliminary notion of inversion first. 
Assuming that $a=_{\dim x.A}b$ is a well-formed type, where $a \in A \dimsub{0}{x}$ and $b \in A \dimsub{1}{x}$, in general the type $b=_{\dim x.A}a$ will \textit{not} be well-formed unless it is also the case that $b \in A \dimsub{0}{x}$ and $a \in A \dimsub{1}{x}$. In other words, the (homogeneous) inversion function from Lemma \ref{lem:inv} fails to be well-defined for every possible well-formed identification type $a=_{\dim x.A}b$, because its inversion operation needs to be subjected to the condition that both $a, b \in A \dimsub{0}{x}$ and $a, b \in A \dimsub{1}{x}$. To put it simply, inverting in the $\dim x$ direction is an operation that only makes sense when $A$ is a degenerate $\dim x$-type, that is, when $\dim x$ does not occur in the type $A$.

Fortunately, there is a way to deal with this problem using type universes.\footnote{We thank Carlo Angiuli and Dan Licata for pointing this out to the author.} Suppose we are given an $\dim x$-type $A$. By dimension abstraction, we have an identification of the types $A \dimsub{0}{x}$ and $A \dimsub{1}{x}$ in a type universe $\univ{U}$ with
$$\dangle{x} A \in A \dimsub{0}{x} =_{\dim x.\univ{U}} A \dimsub{1}{x}.$$
By assumption, the universe $\univ{U}$ must at least be an $\dim x$-type too, nevertheless, we require that $\univ{U}$ be degenerate with respect to $\dim x$, i.e. $\univ{U} \dimsub{0}{x} \equiv \univ{U} \dimsub{1}{x}$. 

Now it can be shown that the following inverse exists:

\vspace{2mm}
\begin{tikzcd}
{} \arrow[r,shorten >= 12pt,-latex,start anchor=center,"\dim x" near start] \arrow[d,shorten >= 9pt,-latex,swap,start anchor=center,"\dim y" near start] & A \dimsub{0}{x}
       \arrow[rrrrrrr,equal] \arrow[dd,"A \equiv (\dangle{x} A)\at y"] &&&&&&& A \dimsub{0}{x} \arrow[dd,equal] \\
 {} &   & \\
&   A \dimsub{1}{x} \arrow[rrrrrrr,dotted, swap,"(\dangle{x} A)^{-1}\at x"] &&&&&&& A \dimsub{0}{x}
\end{tikzcd}
\vspace{2mm}

\noindent In particular, we have
$$(\dangle{x} A)^{-1} \in A \dimsub{1}{x} =_{\dim x.\univ{U}} A \dimsub{0}{x}$$
a construction that gives us an $\dim x$-type $(\dangle{x} A)^{-1}\at x$ (which we shall often abuse notation and write $A^{-1}$). Intuitively, this type corresponds precisely to the `inverse' of the $\dim x$-type $A$. This is because the $- \dimsub{0}{x}$ and $- \dimsub{1}{x}$ faces of $A^{-1}$ are respectively
\begin{align*}
 A^{-1} \dimsub{0}{x}  &\equiv ((\dangle{x} A)^{-1}\at x) \dimsub{0}{x} \\
                       &\equiv (\dangle{x} A)^{-1} \dimsub{0}{x} \at x \dimsub{0}{x} \\
                       &\equiv (\dangle{x} A)^{-1}\at 0  \\
                       &\equiv A \dimsub{1}{x}
\end{align*}
and (similarly) 
\begin{align*}
 A^{-1} \dimsub{1}{x}  &\equiv ((\dangle{x} A)^{-1}\at x) \dimsub{1}{x} \\
                       &\equiv \vdots  \\
                       &\equiv A \dimsub{0}{x}
\end{align*}
which means that we have two inferences that hold top/bottom and bottom/top
$$\Efrac{a \in A^{-1} \dimsub{1}{x}}{a \in A \dimsub{0}{x}} \quad\text{and}\quad \Efrac{a \in A^{-1} \dimsub{0}{x}}{a \in A \dimsub{1}{x}}.$$
Under the assumption that $a=_{\dim x.A}b$ is a well-formed type, it is now easy to see that the type $b=_{\dim x.A^{-1}}a$ will always be well-formed as well regardless of whether $A$ is a degenerate $\dim x$-type or not: because $a=_{\dim x.A}b$ is well-formed we have $a \in A \dimsub{0}{x}$ and $b \in A \dimsub{1}{x}$, meaning that $a \in A^{-1} \dimsub{1}{x}$ and $b \in A^{-1} \dimsub{0}{x}$ must be the case.

This motivates the definition of a new (heterogeneous) inversion operation:

\begin{lemma}[Heterogeneous inversion] \label{lem:hinv}
For every $\dim x$-type $A$ and $a \in A \dimsub{0}{x}$ and $b \in A \dimsub{1}{x}$, there is a function
$$(a=_{\dim x.A}b) \to (b=_{\dim x.A^{-1}}a)$$
called the (heterogeneous) inverse function and denoted $p \mapsto p^{-1_*}$.
\end{lemma}
\begin{proof}
By the following Kan composition on the open square from Lemma \ref{lem:inv}

\vspace{2mm}
\begin{tikzcd}
{} \arrow[r,shorten >= 12pt,-latex,start anchor=center,"\dim x" near start] \arrow[d,shorten >= 9pt,-latex,swap,start anchor=center,"\dim y" near start] & a
       \arrow[rrrrrrrr,equal] \arrow[dd,swap,"p\at y"] &&&&&&&& a \arrow[dd,equal] \\
 {} &   & \\
&   b \arrow[rrrrrrrr,dotted, swap,"{\com^{\lto{0}{1}}_{\dim y.\filler_{\dim y}(A^{-1}\at x)}(a)({}^\dim{0}_\dim{x} \htri \dim y.p \at y, \; {}^\dim{1}_\dim{x} \htri \dim y.a)}"] &&&&&&&& a
\end{tikzcd}
\vspace{2mm}

\noindent where the term 
$$\com^{\lto{r}{s}}_{\dim y.A}(M)({}^\dim{0}_\dim{x_1} \htri \dim y.N_{\dim{x_1}0}, \; {}^\dim{1}_\dim{x_1} \htri \dim y.N_{\dim{x_1}1}, ... , \; {}^\dim{0}_\dim{x_n} \htri \dim y.N_{\dim{x_n}0}, \; {}^\dim{1}_\dim{x_n} \htri \dim y.N_{\dim{x_n}1}),$$
called the \textit{heterogeneous} Kan composite of $M$ with $N_{\dim{x_1}0}$, $N_{\dim{x_1}1}$, ... $N_{\dim{x_n}0}$, $N_{\dim{x_n}1}$, is an abbreviation of the term
\begin{multline*}
\hcom^{\lto{r}{s}}_{\dim A \dimsub{s}{y}}(\coe^{\lto{r}{s}}_{\dim y.A}(M))
({}^\dim{0}_\dim{x_1} \htri \dim y.\coe^{\lto{r}{s}}_{\dim y.A}(N_{\dim{x_1}0}), \; {}^\dim{1}_\dim{x_1} \htri \dim y.\coe^{\lto{r}{s}}_{\dim y.A}(N_{\dim{x_1}1}), ... , \\
\; {}^\dim{0}_\dim{x_n} \htri \dim y.\coe^{\lto{r}{s}}_{\dim y.A}(N_{\dim{x_n}0}), \; {}^\dim{1}_\dim{x_n} \htri \dim y.\coe^{\lto{r}{s}}_{\dim y.A}(N_{\dim{x_n}1})),
\end{multline*}
which combines the two Kan conditions to form a heterogeneous one.
\end{proof}

%%%%

At this point, one may be tempted to think that we can drop our previous definition of (homogeneous) inverse from Lemma \ref{lem:inv}, since we now already possess a more general (heterogeneous) notion of inverse. On second thought, however, it becomes clear that this is not possible on pain of circularity. To put it another way, a preliminary (homogeneous) notion inversion is absolutely necessary in order to define the type $A^{-1}$, so that the definition of heterogeneous inversion is not circular.

Just like homogeneous inversion fails to be well-defined for all possible cases, our (homogeneous) composition function suffers from a similar limitation. This can be easily seen in the two-dimensional case where we have identifications of identifications. Consider the two-dimensional identifications 
$$\alpha \in p =_{\dim y.p' \at y =_{\dim x.A} q' \at y} q \qquad\text{and}\qquad \beta \in q =_{\dim y.q'' \at y =_{\dim x.A} r'' \at y} r$$
which (say) correspond to the following two $(\dim x,\dim y)$-squares in $A$:

\vspace{2mm}
\begin{tikzcd}
{} \arrow[r,shorten >= 12pt,-latex,start anchor=center,"\dim x" near start] \arrow[d,shorten >= 9pt,-latex,swap,start anchor=center,"\dim y" near start] &
    a  \arrow[rrrr,"p\at x"] \arrow[dd,"p' \at  y"] \arrow[ddrrrr, phantom, "\alpha \at y \at x"] &&&& b \arrow[dd,"q' \at  y"] & c \arrow[rrrr,"q\at x"] \arrow[dd,"q'' \at  y"] \arrow[ddrrrr, phantom, "\beta \at y \at x"] &&&& b \arrow[dd,"r'' \at  y"]\\
 {} &   & &&& {} &   & \\
&   c \arrow[rrrr, swap,"q\at x"] &&&& d & e \arrow[rrrr, swap,"r\at x"] &&&& f
\end{tikzcd}
\vspace{2mm}

\noindent it may seem that we can compose $\alpha$ and $\beta$ by `gluing' their common faces together (and composing the other ones) to obtain their composite, as illustrated in the following square:

\vspace{2mm}
\begin{tikzcd}
{} \arrow[r,shorten >= 12pt,-latex,start anchor=center,"\dim x" near start] \arrow[d,shorten >= 9pt,-latex,swap,start anchor=center,"\dim y" near start] & a
       \arrow[rrrrrrrr,"p \at x"] \arrow[dd,swap,"(p' \sq q'')"] \arrow[ddrrrrrrrr, phantom, "\mathsf{glue}_{\alpha,\beta} \at y \at x"] &&&&&&&& b \arrow[dd,"(q' \sq r'') \at y"] \\
 {} &   & \\
&   e \arrow[rrrrrrrr,swap,"r \at x"] &&&&&&&& f
\end{tikzcd}
\vspace{2mm}

\noindent but this composition is actually ill-defined in general because it need not satisfy two essential requirements of Lemma \ref{lem:comp}: first, the target types, $p' \at y =_{\dim x.A} q' \at y$ and $q'' \at y =_{\dim x.A} r'' \at y$, must be degenerate $\dim y$-types; second, those target types must be strictly equal types.

In order to overcome this problem we need a heterogeneous composition operation. Just as with heterogeneous inversion, it be defined with the help of homogeneous composition on types. Once again we assume that $A$ is an $\dim x$-type so that we can obtain an identification of the types $A \dimsub{0}{x}$ and $A \dimsub{1}{x}$ in the universe $\univ{U}$. We also assume that $B$ is a $\dim x$-type such that $A \dimsub{1}{x} \equiv B \dimsub{0}{x}$ and, as before, we require that $\univ{U}$ be a degenerate $\dim x$-type.
As a result, the following homogeneous composition operation is well-defined: 

\vspace{2mm}
\begin{tikzcd}
{} \arrow[r,shorten >= 12pt,-latex,start anchor=center,"\dim x" near start] \arrow[d,shorten >= 9pt,-latex,swap,start anchor=center,"\dim y" near start] & A \dimsub{0}{x}
       \arrow[rrrrrrr,"A \equiv (\dangle{x} A) \at x"] \arrow[dd,swap,equal] &&&&&&& A \dimsub{1}{x} \equiv B \dimsub{0}{x} \arrow[dd,swap,"B \equiv (\dangle{x} B) \at y"] \\
 {} &   & \\
&   A \dimsub{0}{x} \arrow[rrrrrrr,dotted,swap,"(\dangle{x} A \sq \dangle{x} B) \at x"] &&&&&&& B \dimsub{1}{x}
\end{tikzcd}
\vspace{2mm}

\noindent and the composite of $\dangle{x} A$ and $\dangle{x} B$ is the identification
$$ \dangle{x} A \sq \dangle{x} B \; \in \; A \dimsub{0}{x} =_{\dim x.\univ{U}} B \dimsub{1}{x},$$
which can be seen as the $\dim x$-type $(\dangle{x} A \sq \dangle{x} B)\at x$ (henceforth written as $A \sq B$). 
We also have two important inferences that hold top/bottom and bottom/top,
$$\Efrac{a \in (A \sq B) \dimsub{0}{x}}{a \in A \dimsub{0}{x}} \quad\text{and}\quad \Efrac{b \in (A \sq B) \dimsub{1}{x}}{b \in B \dimsub{1}{x}}.$$

With this we have all we need to define our definitive composition function:

%%%%

\begin{lemma}[Heterogeneous composition] \label{lem:hcomp}
Suppose that $A$ and $B$ are $\dim x$-types such that $A \dimsub{1}{x} \equiv B \dimsub{0}{x}$. Given any $a \in A \dimsub{0}{x}$, $b \in A \dimsub{1}{x}$ and $c \in B \dimsub{1}{x}$, there is a function
$$(a=_{\dim x.A}b) \to (b=_{\dim x.B}c) \to (a =_{\dim x.A\,\sq\,B} c)$$
written $p \mapsto q \mapsto p \,\sq_* q$ and called the (heterogeneous) composition function.
\end{lemma}
\begin{proof}
By heterogeneous Kan composition on the open box from Lemma \ref{lem:comp}.
\end{proof}

%%%%

Heterogeneous composition does allow us to compose the two-dimensional identifications $\alpha$ and $\beta$ from our example above, but it is worth noting that the resulting composition is not strictly equal to the operation $\mathsf{glue}_{\alpha,\beta}$ we described (as we shall see in \Cref{idtypegroupoid}, however, this operation is induced by heterogeneous composition.)

%%%%%%%%%%%%%%%%%%%%%%%%%%%%%%%%
\subsection{Heterogeneous groupoid structure} \label{heterogroupstruct}
%%%%%%%%%%%%%%%%%%%%%%%%%%%%%%%%

In \Cref{hgroupstruct} we showed that types have a higher groupoid structure given in terms of homogeneous inversion and composition. The aim of this subsection is to point out that we can characterize this structure via heterogeneous operations as well. To that end, the propositions from \Cref{hgroupstruct} need some adjustments.

To give an illustration let us examine the inversability property (Lemma \ref{lem:inversab}), which in the homogeneous case states that for every degenerate $\dim x$-type $A$ and every $a \in A \dimsub{0}{x}$ and $b \in A \dimsub{1}{x}$ the following holds for any identification $p \in a =_{\dim x.A} b$:
$$(p^{-1})^{-1}=_{\dim y.a =_{\dim x.A} b} p.$$
Generalizing inversability to allow for a heterogeneous operation requires us to first drop the restriction that $A$ be a degenerate $\dim x$-type and then use heterogeneous inversion ${\cdot}^{-1_*}$ to state something like
$$(p^{-1_*})^{-1_*}=_{\dim y.a =_{\dim{x.??}} \, b} p.$$
At this point, however, we run into a problem: on the left-hand side we have a term of type $a =_{\dim x.(A^{-1})^{-1}} b$ but, on the right-hand side, we have a term of type $a=_{\dim x.A} b$. Since, in general, these are not strictly equal $\dim x$-types, we need to find a way to make sure that the above statement is well-typed. 
Fortunately, both $(A^{-1})^{-1}$ and $A$ can be regarded as identifications in a degenerate $\dim x$-type universe $\univ{U}$ as we remarked in the previous section, so we can apply homogeneous inversability (Lemma \ref{lem:inversab}) to obtain a line type from $(A^{-1})^{-1}$ to $A$.

With this in mind, the heterogeneous inversability property can be stated as follows:

\begin{lemma}[Heterogeneous inversability] \label{lem:hinversab}
For every $\dim x$-type $A$ with $a \in A \dimsub{0}{x}$ and $b \in A \dimsub{1}{x}$, we have
$$(p^{-1_*})^{-1_*}=_{\dim y.a =_{\dim x.(\mathsf{inv}_{A}) \at x} b} p$$
for any $p \in a =_{\dim x.A} b$.
\end{lemma}

The argument is just like the proofs of our heterogeneous generalizations of inversion (Lemma \ref{lem:hinv}) and composition (Lemma \ref{lem:hcomp}): a straightforward heterogeneous Kan composition on the open cube constructed for the proof of its homogeneous counterpart (Lemma \ref{lem:inversab}). In fact, all heterogeneous counterparts of the propositions from \Cref{hgroupstruct} follow the same pattern (they can all be stated by using their homogeneous counterparts and proven by a heterogeneous filling of their open cubes), so we will simply omit those results.

%%%%%%%%%%%%%%%%%%%%%%%%%%%%%%%%
\section{General notable properties of identifications} \label{moreabout}
%%%%%%%%%%%%%%%%%%%%%%%%%%%%%%%%

In this section we explore a few notable properties of cubical identifications including path induction, properties of loops and some peculiarities of the groupoid operations applied to identification types.

%%%%%%%%%%%%%%%%%%%%%%%%%%%%%%%%
\subsection{Path induction} \label{pathind}
%%%%%%%%%%%%%%%%%%%%%%%%%%%%%%%%

We start with path induction (otherwise known as $J$), a fundamental property that serves as the elimination rule of the identification type in both standard constructive type theory and conventional homotopy type theory \cite[\S1.12.1]{hottbook}. Roughly, path induction states that identifications (paths) can be deformed and retracted without changing their essential characteristics. 

This can be expressed cubically as follows.

\begin{theorem}[Path induction]
Given an $\dim x$-type $A$, a term $a \in A \dimsub{0}{x}$ and a type family $P \in \prod_{(x \in A \dimsub{1}{x})} (a =_{\dim x.A} x) \to \univ{U}$ we have a function
$$\mathsf{J} \in \prod_{(x \in A \dimsub{1}{x})} \prod_{(p \in a =_{\dim x.A} x)} \prod_{(u \in P(a,\refl_a))} P(x,p)$$
\end{theorem}
\begin{proof}
Suppose we are given $x \in A \dimsub{1}{x}$, $p \in a =_{\dim x.A} x$ and $u \in P(a,\refl_a)$. The idea of the following proof is to coerce $u$ over a line between the types $P(a,\refl_a)$ and $P(x,p)$, so we shall focus on the construction of this type line first.

We construct the following $(\dim x,\dim y)$-square by homogeneous Kan composition

\vspace{2mm}
\begin{tikzcd}
{} \arrow[r,shorten >= 12pt,-latex,start anchor=center,"\dim x" near start] \arrow[d,shorten >= 9pt,-latex,swap,start anchor=center,"\dim y" near start] & a
       \arrow[rrrrrrrr,equal] \arrow[dd,equal] \arrow[ddrrrrrrrr,phantom,"\mathsf{is\_refl}_p \at y \at x"] &&&&&&&& a \arrow[dd,"p\at y"] \\
 {} &   & \\
&   a \arrow[rrrrrrrr, swap,"p\at x"] &&&&&&&& x
\end{tikzcd}
\vspace{2mm}

\noindent using the filler from the composites of Lemmas \ref{lem:comp} and \ref{lem:ru} (the argument is similar to the one given in our proof of Lemma \ref{lem:op} (i)).

Our next step is to observe that $\mathsf{is\_refl}_p$ induces an $\dim x$-line in $\univ{U}$
$$P(p \at x, \mathsf{is\_refl}_p\at x)$$
from $P(a,\refl_a) \in \univ{U} \dimsub{0}{x}$ to $P(x,p) \in \univ{U} \dimsub{0}{x}$ (as required), since we have
\begin{align*}
 P(p \at x, \mathsf{is\_refl}_p\at x) \dimsub{0}{x}  &\equiv P(p \at 0,\mathsf{is\_refl}_p \at 0) \\
                       &\equiv P(a,\refl_a)
\end{align*}
and, in a similar fashion,
\begin{align*}
 P(p \at x, \mathsf{is\_refl}_p\at x) \dimsub{1}{x}  &\equiv P(p \at 1,\mathsf{is\_refl}_p \at 1) \\
                       &\equiv P(x,p).
\end{align*}

To complete the proof we just need to coerce $u \in P(a,\refl_a)$ on this line, so
$$ \mathsf{J} :\equiv \fun x. \fun p. \fun u. \coe^{\lto{0}{1}}_{\dim{x}.P(p \at x, \mathsf{is\_refl}_p\at x)} (u)$$
gives the required function.
\end{proof}

Since the above principle assumes that identifications must always have one of its endpoints predetermined, this property is sometimes regarded as special form of path induction called \emph{based path induction} in the literature \cite{hottbook}[\S1.12.1].

%%%%%%%%%%%%%%%%%%%%%%%%%%%%%%%%
\subsection{Loops}
%%%%%%%%%%%%%%%%%%%%%%%%%%%%%%%%

In conventional homotopy type theory, loops are identifications with the same start and end points (up to strict equality) \cite[\S2]{hottbook}. This characterization is very convenient because there is no natural way of expressing non-globular (cubical) identifications in conventional homotopy type theory and, consequently, every identification must have exactly two endpoints. Since globular identifications are just a particular sort of cubical identifications, we may view loops as globular identifications with strictly equal endpoints.

Thus, one-dimensional loops are the inhabitants of the type $a =_{\dim x.A} a$ (the loop space of $a$), as usual \cite[\S2]{hottbook}, but two-dimensional loops, that is, the terms of the type $p =_{\dim y.{a =_{\dim x.A}b}} p$ (the loop space of $p$), are homogeneous degenerate squares and so on. 
Cubically, the loop space of the loop space of $a$ is represented by the type $\refl_a =_{\dim y.a=_{\dim x.A}a} \refl_a$ just as in conventional homotopy type theory \cite[\S2.1]{hottbook} and composition of loops is commutative as well:

\begin{theorem}[Eckmann-Hilton]
Given any degenerate $\dim x$-type $A$ and term $a \in A \dimsub {0}{x}$, the following homogeneous composition is commutative
$$\alpha \sq \beta =_{\dim z. a =_{\dim y.a=_{\dim x.A} a} a} \beta \sq \alpha$$
for $\alpha, \beta \in \refl_a =_{\dim y.a=_{\dim x.A}a} \refl_a$.
\end{theorem}
\begin{proof}
By a cubical simplification of the proof of Theorem 2.1.6 in \cite{hottbook}.

Right whiskering is an operation that given a two-dimensional identification $\alpha \in p =_{\dim y. a =_{\dim x.A} b} q$ and a one-dimensional identification $r \in b =_{\dim x.A} c$, returns a term of type $p \sq r =_{\dim y. a =_{\dim x.A} b} q \sq r $. In conventional homotopy type theory this is an operation that requires definition \cite[\S2.1]{hottbook}. Cubically, however, it turns out that right whiskering is just a particular instance of homogeneous composition because degeneration ensures that the composite identification $\alpha \sq r$ will always be a well-formed term, as can be seen below:

\vspace{2mm}
\begin{tikzcd}
{} \arrow[r,shorten >= 12pt,-latex,start anchor=center,"\dim x" near start] \arrow[d,shorten >= 9pt,-latex,swap,start anchor=center,"\dim y" near start] &
    a  \arrow[rrrr,"p\at x"] \arrow[dd,equal] \arrow[ddrrrr, phantom, "\alpha \at y \at x"] &&&& b \arrow[dd,equal] & b \arrow[rrrr,"r\at x"] \arrow[dd,equal] \arrow[ddrrrr, phantom, "r"] &&&& c \arrow[dd,equal]\\
 {} &   & &&& {} &   & \\
&   a \arrow[rrrr, swap,"q\at x"] &&&& b & b \arrow[rrrr,swap,"r\at x"] &&&& c
\end{tikzcd}
\vspace{2mm}

\noindent Naturally, the same holds for left whiskering, which states that for any one-dimensional identification $p \in a =_{\dim x.A} b$ and two-dimensional identification $\beta \in r =_{\dim y. b =_{\dim x.A} c} s$ we have an inhabitant of the type $p \sq r =_{\dim y. a =_{\dim x.A} b} p \sq s $.

It is easy to see by path induction on the abovementioned identifications $\alpha$, $\beta$, $p$ and $r$ that whiskering is commutative, that is, the following is true
$$
(\alpha \sq r) \sq (q \sq \beta) =_{\dim y. a =_{\dim x.A} c}  (p \sq \beta) \sq (\alpha \sq s).
$$
But now we note that the above proposition already shows that $\alpha \sq \beta = \beta \sq \alpha$ always holds when $\alpha, \beta \in \refl_a =_{\dim y.a=_{\dim x.A}a} \refl_a$ because the reflexivity element is both a right and left unit for composition (see Lemmas \ref{lem:ru} and \ref{lem:lu}).
\end{proof}

%%

%%%%%%%%%%%%%%%%%%%%%%%%%%%%%%%%
\subsection{Groupoid operations on identification types} \label{idtypegroupoid}
%%%%%%%%%%%%%%%%%%%%%%%%%%%%%%%%

We conclude this paper with a few remarks and results about how the groupoid structure holds for identification types (seen not as types but identifications).

How can we view types as identifications? Recall from \Cref{hetgroupoid} that when $A$ is an $\dim x$-type in a type universe $\mathcal{U}$ we have an identification
$$A \dimsub{0}{x} =_{\dim x.\univ{U}} A \dimsub{1}{x}$$
given by the term $\dangle{x}A$ (which to avoid pedantism we simply write as $A$). Thus, if we consider an identification $\dim x$-type like $p \at x =_{\dim y.A} q \at x$, given the usual assumptions that $p \in a=_{\dim x.A \dimsub{0}{y}} b$ and $q \in c=_{\dim x.A \dimsub{1}{y}} d$ and so on (see Theorem \ref{idistr} below), then we may see this type as an identification in the universe $\univ{U}$ between the types
$$p \at x =_{\dim y.A} q \at x \dimsub{0}{x} \equiv a=_{\dim y.A \dimsub{0}{x}}c$$ 
and 
$$p \at x =_{\dim y.A} q \at x \dimsub{1}{x} \equiv b=_{\dim y.A \dimsub{1}{x}}d.$$
In other words, we may view the identification $\dim x$-type $p \at x =_{\dim y.A} q \at x$ as an identification of identification types inhabiting the type
$$(a =_{\dim y.A \dimsub{0}{x}} c)  =_{\dim x.\univ{U}}  (b =_{\dim y.A \dimsub{1}{x}}d).$$
The moral of the story is that, since identification types may be regarded as identifications, they are also subject to the groupoid operations like all identifications are by default.

In fact, we can show by path induction that inversion can be distributed over the identification type:

%%%%

\begin{theorem}[Identification type inversion distribution] \label{idistr}
Suppose that $A \in \univ{U}$ is an $(\dim x, \dim y)$-type and $p \in a=_{\dim x.A \dimsub{0}{y}} b$ and $q \in c=_{\dim x.A \dimsub{1}{y}} d$. We have
$$(p \at x =_{\dim y.A} q \at x)^{-1} =_{\dim y.(b =_{\dim y. A \dimsub{1}{x}} d) =_{\dim x.\univ{U}} (a =_{\dim y. A \dimsub{0}{x}} c)} (p^{-1} \at x =_{\dim y.A^{-1}} q^{-1} \at x)$$
where $a \in A\dimsub{0}{x}\dimsub{0}{y}$, $b \in A\dimsub{1}{x}\dimsub{0}{y}$, $c \in A\dimsub{0}{x}\dimsub{1}{y}$, $d \in A\dimsub{1}{x}\dimsub{1}{y}$.
\end{theorem}
\begin{proof} 
The statement of this theorem can be expressed as the expectation of the construction of the following square

\vspace{2mm}
\begin{tikzcd}
{} \arrow[r,shorten >= 12pt,-latex,start anchor=center,"\dim x" near start] \arrow[d,shorten >= 9pt,-latex,swap,start anchor=center,"\dim y" near start] & b =_{\dim y. A \dimsub{1}{x}} d
       \arrow[rrrrrrr,"(p \at x =_{\dim y.A} q \at x)^{-1} \at x"] \arrow[dd,equal] \arrow[ddrrrrrrr, phantom,"\mathsf{id}_{p,q} \at y \at x"] &&&&&&& a =_{\dim y. A \dimsub{0}{x}} c
        \arrow[dd, equal] \\
 {} &   & \\
& b =_{\dim y. A \dimsub{1}{x}} d \arrow[rrrrrrr, swap,"p^{-1} \at x =_{\dim y.(\dangle{x} A^{-1}) \at x} q^{-1} \at x"] &&&&&&& a =_{\dim y. A \dimsub{0}{x}} c
\end{tikzcd}
\vspace{2mm}

\noindent which we shall call $\mathsf{id}_{p,q} \at y \at x$. To that end, we proceed by triple path induction.

First we do path induction on $A$, which allows us to assume that $A$ is $\refl_A$ (in other words, we can suppose that $A$ is a degenerate $(\dim x, \dim y)$-type). By induction on $p$, it suffices to assume that $b$ is $a$ and that $p$ is $\refl_a$. Again, by induction on $q$, it suffices to assume also that $d$ is $c$ and $q$ is $\refl_c$.

After all that, the above square reduces to the following:

\vspace{2mm}
\begin{tikzcd}
{} \arrow[r,shorten >= 12pt,-latex,start anchor=center,"\dim x" near start] \arrow[d,shorten >= 9pt,-latex,swap,start anchor=center,"\dim y" near start] & a =_{\dim y. A \dimsub{0}{x}} c
       \arrow[rrrrrrr,"(\refl_a \at x =_{\dim y.A} \refl_c \at x)^{-1} \at x"] \arrow[dd,equal] \arrow[ddrrrrrrr, phantom,] &&&&&&& a =_{\dim y. A \dimsub{0}{x}} c
        \arrow[dd,equal] \\
 {} &   & \\
& a =_{\dim y. A \dimsub{0}{x}} c \arrow[rrrrrrr, swap,"(\refl_a)^{-1} \at x =_{\dim y.(\refl_A)^{-1} \at x} (\refl_c)^{-1} \at x"] &&&&&&& a =_{\dim y. A \dimsub{0}{x}} c
\end{tikzcd}
\vspace{2mm}

\noindent However, recall that the reflexivity element equals its inverse up to globular identification (Lemma \ref{lem:invu}), so the term (which we abbreviate by \textsf{U})
$$\mathsf{iu}_a \at y \at x =_{\dim y.\mathsf{iu}_A \at y \at x} \mathsf{iu}_c \at y \at x$$
is an $(\dim x, \dim y)$-square 

\vspace{2mm}
\begin{tikzcd}
{} \arrow[r,shorten >= 12pt,-latex,start anchor=center,"\dim x" near start] \arrow[d,shorten >= 9pt,-latex,swap,start anchor=center,"\dim y" near start] &
    a =_{\dim y.A \dimsub{0}{x}} c  \arrow[rrrrrrr,equal,"\refl_a \at x =_{\dim y.A} \refl_c \at x"] \arrow[dd,equal] \arrow[ddrrrrrrr, phantom, "\mathsf{U} :\equiv \mathsf{iu}_a \at y \at x =_{\dim y.\mathsf{iu}_A \at y \at x} \mathsf{iu}_c \at y \at x"] &&&&&&& a =_{\dim y.A \dimsub{0}{x}} c \arrow[dd,equal] \\
 {} &   & \\
&   a =_{\dim y.A \dimsub{0}{x}} c \arrow[rrrrrrr, swap,"(\refl_a)^{-1} \at x =_{\dim y.(\refl_A)^{-1} \at x} (\refl_c)^{-1} \at x"] &&&&&&& a =_{\dim y.A \dimsub{0}{x}} c
\end{tikzcd}
\vspace{2mm}

\noindent since the top (1-2), bottom (3-4), and left and right (5-6) faces of \textsf{U} are

\begin{align}
\mathsf{U} \dimsub{0}{y} &\equiv
  \mathsf{iu}_a \at 0 \at x =_{\dim y.\mathsf{iu}_A \at 0 \at x} \mathsf{iu}_c \at 0 \at x \\
  &\equiv \refl_a \at x =_{\dim y.A} \refl_c \at x,
\end{align}

\begin{align}
\mathsf{U} \dimsub{1}{y} &\equiv \mathsf{iu}_a \at 1 \at x =_{\dim y.\mathsf{iu}_A \at 1 \at x} \mathsf{iu}_c \at 1 \at x \\
   &\equiv (\refl_a)^{-1} \at x =_{\dim y.(\refl_A)^{-1}\at x} (\refl_c)^{-1} \at x,
\end{align}

and

\begin{align}
\mathsf{U} \dimsub{\epsilon}{x} &\equiv
  \refl_a \at x =_{\dim y.A} \refl_c \at x \\
  &\equiv a =_{\dim y.A} c. 
\end{align}

Finally, we observe that the top face of \textsf{U} is a degenerate line because
$$(\refl_a \at x =_{\dim y.A} \refl_c \at x) \equiv \refl_{a =_{\dim y.A} c}$$
so we apply Lemma \ref{lem:invu} again to obtain a new $(\dim x, \dim y)$-square \textsf{I}:

\vspace{2mm}
\begin{tikzcd}
{} \arrow[r,shorten >= 12pt,-latex,start anchor=center,"\dim x" near start] \arrow[d,shorten >= 9pt,-latex,swap,start anchor=center,"\dim y" near start] &
    a =_{\dim y.A \dimsub{0}{x}} c  \arrow[rrrrrrr,"(\refl_a \at x =_{\dim y.A} \refl_c \at x)^{-1} \at x"] \arrow[dd,equal] \arrow[ddrrrrrrr, phantom, "I :\equiv (\mathsf{iu}_{a =_{\dim y.A} c})^{-1} \at y \at x"] &&&&&&& a =_{\dim y.A \dimsub{0}{x}} c \arrow[dd,equal] \\
 {} &   & \\
&   a =_{\dim y. A \dimsub{0}{x}} c \arrow[rrrrrrr,swap,equal,"\refl_a \at x =_{\dim y.A} \refl_c \at x"] &&&&&&& a =_{\dim y.A \dimsub{0}{x}} c
\end{tikzcd}
\vspace{2mm}

\noindent We complete the proof by letting $\mathsf{id}_{p,q}$ be the homogeneous composition of the squares $\mathsf{I}$ and $\mathsf{U}$ (regarded as identifications).

\end{proof}

%%%%

It is an interesting fact that when (homogeneous) inversion is well-defined for a two-dimensional identification such as $\alpha \in p=_{\dim{y}.r \at y =_{\dim x.A} s \at y}q$, both the original identification $\alpha$ and its inversion $\alpha^{-1}$ may be regarded as squares in $A$. Curiously enough, despite the general applicability of heterogeneous inversion, the same cannot be said for the heterogeneous inversion of $\alpha$
$$\alpha^{-1_*} \in q=_{\dim{y}.(r \at y =_{\dim x.A} s \at y)^{-1}}p,$$
which can only be pictured as a line from $q$ to $p$ in the inverted identification type $(r \at y =_{\dim x.A} s \at y)^{-1}$. The reason is clear: after the inversion of an identification type the resulting type stops being an identification type.

Fortunately, because the above theorem tells us that inversion can always be distributed over the identification type, the heterogeneous inversion of any square in $A$ equals a swapped square in $A^{-1}$ (in the sense of Lemma \ref{lem:swap}) up to identification.

%%%%

\begin{corollary}[Heterogeneous square swap] \label{lem:hswap}
For every $(\dim x, \dim y)$-type $A \in \univ{U}$ and every $p \in a =_{\dim x.A \dimsub{0}{y}} b$, $q \in c =_{\dim x.A \dimsub{1}{y}} d$, $r \in a =_{\dim y.A \dimsub{0}{x}} c$, $s \in b =_{\dim y.A \dimsub{1}{x}} d$, we have
$$\mathsf{hswap}_\alpha \in (p=_{\dim{y}.(r \at y =_{\dim x.A} s \at y)^{-1}}q) =_{\dim x.\univ{U}} (p=_{\dim{y}.{r}^{-1} \at y =_{\dim x.A^{-1}} {s}^{-1} \at y}q)$$
where $a \in A\dimsub{0}{x}\dimsub{0}{y}$, $b \in A\dimsub{1}{x}\dimsub{0}{y}$, $c \in A\dimsub{0}{x}\dimsub{1}{y}$, $d \in A\dimsub{1}{x}\dimsub{1}{y}$.
\end{corollary}
\begin{proof}
By Theorem \ref{idistr} we have an identification between the $\dim y$-types $(r \at y =_{\dim x.A} s \at y)^{-1}$ and ${r}^{-1} \at y =_{\dim x.A^{-1}} {s}^{-1} \at y$, so
$$p=_{\dim{y}.\mathsf{id}_{r,s} \at x} q$$
is an $\dim x$-line in $\univ{U}$ from $p=_{\dim{y}.(r \at y =_{\dim x.A} {s})^{-1} \at y}q$ to $p=_{\dim{y}.{r}^{-1} \at y =_{\dim x.A^{-1}} {s}^{-1} \at y}q$.
\end{proof}

%%%%

We can show by path induction that composition can be distributed over the identification type as well. We enunciate this as follows:

\begin{theorem}[Identification type composition distribution] \label{compdistr}
Given $(\dim x, \dim y)$-types $A, B \in \univ{U}$ such that $A\dimsub{1}{x} \equiv B \dimsub{0}{x}$ and $p \in a=_{\dim x.A \dimsub{0}{y}}b$ and $q \in c=_{\dim x.A \dimsub{1}{y}}d$, and $r \in b=_{\dim x.B \dimsub{0}{y}}e$ and $s \in d=_{\dim x.B \dimsub{1}{y}}f$, we have
$$(p \at x =_{\dim y.A} q \at x) \sq (r \at x =_{\dim y.B} s \at x) =_{\dim y.(a =_{\dim y. A \dimsub{1}{x}} c) =_{\dim x.\univ{U}} (e =_{\dim y. A \dimsub{0}{x}} f)} ((p \,\sq\, r) \at x =_{\dim y.A \,\sq\, B} (q \,\sq\, s) \at x)$$
where $a \in A\dimsub{0}{x}\dimsub{0}{y}$, $b \in A\dimsub{1}{x}\dimsub{0}{y}$, $c \in A\dimsub{0}{x}\dimsub{1}{y}$, $d \in A\dimsub{1}{x}\dimsub{1}{y} $, $e \in B\dimsub{0}{x}\dimsub{1}{y}$, $f \in B\dimsub{1}{x}\dimsub{1}{y}$.
\end{theorem}
\begin{proof} 
The proof follows the same idea as the proof of Theorem \ref{idistr}, so we shall skip the details. 
We just note that, by path induction, it suffices to find the following square

\vspace{2mm}
\begin{tikzcd}
{} \arrow[r,shorten >= 12pt,-latex,start anchor=center,"\dim x" near start] \arrow[d,shorten >= 9pt,-latex,swap,start anchor=center,"\dim y" near start] & a =_{\dim y. A \dimsub{0}{x}} c
       \arrow[rrrrrrr,"(\refl_a \at x =_{\dim y.A} \refl_c \at x) \sq (\refl_a \at x =_{\dim y.A} \refl_c \at x) \at x"] \arrow[dd,equal] \arrow[ddrrrrrrr, phantom,] &&&&&&& a =_{\dim y. A \dimsub{0}{x}} c
        \arrow[dd, equal] \\
 {} &   & \\
& a =_{\dim y. A \dimsub{0}{x}} c \arrow[rrrrrrr, swap,"((\refl_a \sq \refl_a) \at x =_{\dim y.A \sq A} (\refl_c \sq \refl_c)\at x) \at x"] &&&&&&& a =_{\dim y. A \dimsub{0}{x}} c
\end{tikzcd}
\vspace{2mm}

\noindent which can be easily constructed using Lemma \ref{lem:compu} in a similar way to that of Theorem \ref{idistr}.
\end{proof}

This theorem shows that the gluing operation described in \Cref{hetgroupoid} is closely related to homogeneous composition. Put differently, it shows that the homogeneous composition of two squares induces a gluing operation:

\begin{corollary}[Heterogeneous square gluing] \label{lem:hglue}
Suppose that $A, B \in \univ{U}$ are $(\dim x, \dim y)$-types such that $A\dimsub{1}{y} \equiv B \dimsub{0}{y}$. For every $p \in a =_{\dim x.A \dimsub{0}{y}} b$, $q \in c =_{\dim x.A \dimsub{1}{y}} d$, $r \in a =_{\dim y.A \dimsub{0}{x}} c$, $s \in b =_{\dim y.A \dimsub{1}{x}} d$, $t \in c =_{\dim y.B \dimsub{0}{x}} e$, $u \in d =_{\dim y.B \dimsub{1}{x}} f$, the following holds
$$\mathsf{hglue}_\alpha \in (p=_{\dim{y}.(r \at y =_{\dim x.A} s \at y) \sq (t \at y =_{\dim x.B} u \at y)}q) =_{\dim x.\univ{U}} (p=_{\dim{y}.(r \,\sq\, t) \at y =_{\dim x.A \,\sq\, B} (s \,\sq\, u) \at y}q)$$
where $a \in A\dimsub{0}{x}\dimsub{0}{y}$, $b \in A\dimsub{1}{x}\dimsub{0}{y}$, $c \in A\dimsub{0}{x}\dimsub{1}{y}$, $d \in A\dimsub{1}{x}\dimsub{1}{y} $, $e \in B\dimsub{0}{x}\dimsub{1}{y}$, $f \in B\dimsub{1}{x}\dimsub{1}{y}$.
\end{corollary}
\begin{proof}
In a manner similar to our proof of Corollary \ref{lem:hswap}, this can be proven via the identification given by Theorem \ref{compdistr}.
\end{proof}

Finally, we note that the distribution of inversion and composition (as in Theorems \ref{idistr} and \ref{compdistr}) over the identification type give us a new characterization of the groupoid laws: when identification types are taken as identifications we can characterize the groupoid structure via distribution as follows.

\begin{corollary}[Identification type groupoid laws] \label{idistr}
Suppose that $A$, $B$ and $C$ are $(\dim x, \dim y)$-types such that $A \dimsub{1}{x} \equiv B \dimsub{0}{x}$ and $B \dimsub{1}{x} \equiv C \dimsub{0}{x}$. Given any $p \in a=_{\dim x.A \dimsub{0}{y}} b$, $q \in c=_{\dim x.A \dimsub{1}{y}} d$, $r \in b=_{\dim x.B \dimsub{0}{y}}e$, $s \in d=_{\dim x.B \dimsub{1}{y}}f$, $t \in e=_{\dim x.C \dimsub{0}{y}}g$ and $u \in f=_{\dim x.C \dimsub{1}{y}}h$, we have the following:
\begin{enumerate}[(i)]
	\item $(p^{-1} \at x =_{\dim y.A^{-1}} q^{-1} \at x)^{-1} = (p \at x =_{\dim y.A} q \at x)$ %$(p^{-1} \at x =_{\dim x.A^{-1}} q^{-1} \at x)^{-1} =_{\dim y.(a =_{\dim y. A \dimsub{0}{x}} c) =_{\dim x.\univ{U}} (b =_{\dim y. A \dimsub{1}{x}} d)} p \at x =_{\dim x.A} q \at x$
	\item $(p \at x =_{\dim y.A} q \at x) \sq (p^{-1} \at x =_{\dim y.A^{-1}} q^{-1} \at x) = \refl_{(a =_{\dim y. A \dimsub{0}{x}} c)}$ %$(p \at x =_{\dim x.A} q \at x) \sq (p^{-1} \at x =_{\dim x.A^{-1}} q^{-1} \at x) =_{\dim y.(a =_{\dim y. A \dimsub{0}{x}} c) =_{\dim x.\univ{U}} (a =_{\dim y. A \dimsub{0}{x}} c)} \refl_{(a =_{\dim y. A \dimsub{0}{x}} c)}$
		\item $(p^{-1} \at x =_{\dim y.A^{-1}} q^{-1} \at x) \sq (p \at x =_{\dim y.A} q \at x) = \refl_{(b =_{\dim y. A \dimsub{1}{x}} d)}$ %$(p^{-1} \at x =_{\dim x.A^{-1}} q^{-1} \at x) \sq (p \at x =_{\dim x.A} q \at x) =_{\dim y.(b =_{\dim y. A \dimsub{0}{x}} d) =_{\dim x.\univ{U}} (b =_{\dim y. A \dimsub{0}{x}} d)} \refl_{(b =_{\dim y. A \dimsub{1}{x}} d)}$
		\item $(p \at x =_{\dim y.A} q \at x) \sq \refl_{(b =_{\dim y. A \dimsub{1}{x}} d)} = (p \at x =_{\dim y.A} q \at x)$
		\item $\refl_{(a =_{\dim y. A \dimsub{0}{x}} c)} \sq (p \at x =_{\dim y.A} q \at x) = (p \at x =_{\dim y.A} q \at x)$
		\item $\begin{multlined}[t]
((p \at x =_{\dim y.A} q \at x) \sq (r \at x =_{\dim y.B} s \at x)) \sq (t \at x =_{\dim y.C} u \at x) = \\
(p \at x =_{\dim y.A} q \at x) \sq ((r \at x =_{\dim y.B} s \at x) \sq (t \at x =_{\dim y.C} u \at x))
\end{multlined}$ \label{iassoc} %$((p \at x =_{\dim x.A} q \at x) \sq (r \at x =_{\dim x.A} s \at x)) \sq (t \at x =_{\dim x.A} u \at x) = (p \at x =_{\dim x.A} q \at x) \sq ((r \at x =_{\dim x.A} s \at x) \sq (t \at x =_{\dim x.A} u \at x))$
\end{enumerate}
where $a \in A\dimsub{0}{x}\dimsub{0}{y}$, $b \in A\dimsub{1}{x}\dimsub{0}{y}$, $c \in A\dimsub{0}{x}\dimsub{1}{y}$, $d \in A\dimsub{1}{x}\dimsub{1}{y}$, $e \in B \dimsub{1}{x}\dimsub{0}{y}$, $f \in B \dimsub{1}{x}\dimsub{1}{y}$, $g \in C \dimsub{1}{x}\dimsub{0}{y}$ and $h \in C \dimsub{1}{x} \dimsub{1}{y}$.
\end{corollary}
\begin{proof}
Routine use of (i) inversability (Lemma \ref{lem:inversab}), (ii) right cancellation (Lemma \ref{lem:rc}), (iii) left cancellation (Lemma \ref{lem:lc}), (iv) right unit (Lemma \ref{lem:ru}), (v) left unit (Lemma \ref{lem:lu}), (\ref{iassoc}) associativity (Lemma \ref{lem:assoc}) and distribution of inversion (Theorems \ref{idistr}) and composition (\ref{compdistr}) over the identification type in (i-iii) and (vi). 
\end{proof}

In the above corollary the type indexes of the outermost identification type were removed for the sake of brevity (at this point we trust that the reader is able to easily obtain this information by checking the types of the terms involved in the both sides of the expression).

%%%%

%%%%%%%%%%%%%%%%%%%%%%%%%%%%%%%%
\section{Future work}
%%%%%%%%%%%%%%%%%%%%%%%%%%%%%%%%

There is much to be done yet in order to provide a cubical alternative to the informal type theory project of the homotopy type theory book \cite{hottbook}. We view this paper as opening up many possibilities for future work, including informal cubical accounts of the functorality of functions, fibrational aspect of type families, the higher groupoid structure of type formers (including function extensionality and univalence), higher inductive types, homotopy $n$-types and so on.

Another important area for future work is the mechanization of the proofs presented in this paper using the young proof assistant RedPRL \cite{redprl}. The type theory of RedPRL contains additional type formers not used in this paper such as \emph{line types}, which are identification types with arbitrary endpoints. Despite its practical advantages, however, line types have the same expressive power as identification types.
Part of the proofs contained in this paper have already been formalized using line types and are available online.\footnote{\url{https://github.com/RedPRL/sml-redprl/blob/master/example/groupoid.prl}}

%%%%%%%%%%%%%%%%%%%%%%%%%%%%%%%%
%% Acknowledgements
%%%%%%%%%%%%%%%%%%%%%%%%%%%%%%%%

\vspace{5mm}
\noindent \textbf{Acknowledgments} \; The author wishes to thank Robert Harper, Carlo Angiuli and Dan Licata for invaluable conversations on the topic covered herein. The author is also indebted to Carlo Angiuli and Thierry Coquand for helpful comments on an earlier draft of this paper. All mistakes are the authors' own.

%%%%%%%%%%%%%%%%%%%%%%%%%%%%%%%%
%% References
%%%%%%%%%%%%%%%%%%%%%%%%%%%%%%%%

\bibliography{ref}

\begin{thebibliography}{10}

\bibitem{ah17}
Carlo Angiuli and Robert Harper.
\newblock {Computational Higher Type Theory II: Dependent Cubical
  Realizability}.
\newblock URL: \url{https://arxiv.org/abs/1606.09638}, 4 2017.
\newblock Preprint.

\bibitem{ahx}
Carlo Angiuli and Robert Harper.
\newblock {Meaning explanations at higher dimension}.
\newblock URL: \url{http://www.cs.cmu.edu/~cangiuli/papers/brouwer.pdf}, 11
  2018.
\newblock To appear in the special issue \emph{L.E.J. Brouwer, fifty years
  later}.

\bibitem{ang17}
Carlo Angiuli, Robert Harper, and Todd Wilson.
\newblock {Computational Higher-Dimensional Type Theory}, 2017.

\bibitem{chtt3}
Carlo Angiuli, Kuen-Bang {Hou (Favonia)}, and Robert Harper.
\newblock {Computational Higher Type Theory III: Univalent Universes and Exact
  Equality}.
\newblock URL: \url{https://arxiv.org/pdf/1712.01800.pdf}, 12 2017.
\newblock Preprint.

\bibitem{aw09}
Steve Awodey and Micheal Warren.
\newblock {Homotopy theoretic models of identity types}.
\newblock {\em Mathematical Proceedings of the Cambridge Philosophical
  Society}, 146:45--55, 2009.

\bibitem{bch14}
Marc Bezem, Thierry Coquand, and Simon Huber.
\newblock {A model of type theory in cubical sets}.
\newblock {\em 19th International Conference on Types for Proofs and Programs
  (TYPES 2013)}, 26:107--128, 2014.

\bibitem{bmx}
Ulrik Buchholtz and Edward Morehouse.
\newblock {Varieties of Cubical Sets}.
\newblock URL: \url{https://arxiv.org/abs/1701.08189}, 2017.

\bibitem{cch16}
Cyril Cohen, Thierry Coquand, and Simon Huber.
\newblock {Cubical type theory: a constructive interpretation of the univalence
  axiom}.
\newblock Preprint, 2016.

\bibitem{kan55}
Daniel~M. Kan.
\newblock {Abstract homotopy. i}.
\newblock {\em Proceedings of the National Academy of Sciences of the United
  States of America}, 41(12):1092--1096, 1955.

\bibitem{lb14}
Dan~R. Licata and Guillaume Brunerie.
\newblock {A cubical type theory}.
\newblock URL:
  \url{http://dlicata.web.wesleyan.edu/pubs/lb14cubical/lb14cubes-oxford.pdf},
  11 2014.
\newblock Talk at Oxford Homotopy Type Theory Workshop.

\bibitem{ml75}
Per Martin-L{\"o}f.
\newblock An intuitionistic theory of types: predicative part.
\newblock In H.~E. Rose and J.~C. Shepherdson, editors, {\em Logic Colloquium
  ’73 : Proceedings of the logic colloquium, Bristol}, pages 73--118.
  North-Holland, Amsterdam, New York, Oxford, 7 1975.

\bibitem{shu13}
Micheal Shulman.
\newblock {The HoTT Book}.
\newblock URL:
  \url{https://golem.ph.utexas.edu/category/2013/06/the_hott_book.html}, 2013.
\newblock Online. Accessed: 2018-02-16.

\bibitem{redprl}
Jonathan Sterling, Kuen-Bang {Hou (Favonia)}, Evan Cavallo, James Wilcox,
  Eugene Akentyev, David Christiansen, Daniel Gratzer, and Darin Morrison.
\newblock {RedPRL\textemdash the People's Refinement Logic}.
\newblock URL: \url{http://www.redprl.org/}, 2018.
\newblock Online. Accessed: 2018-02-22.

\bibitem{hottbook}
{The Univalent Foundations Program}.
\newblock Homotopy type theory: Univalent foundations of mathematics, 2013.

\end{thebibliography}
\bibliographystyle{plain}

\end{document}